\RequirePackage{fix-cm}
\documentclass[twocolumn]{svjour3}          
\smartqed  
\usepackage{amssymb,amsfonts,amsmath,latexsym,graphicx,url,moreverb,xcolor}
\usepackage{natbib}
\usepackage[utf8]{inputenc}
\usepackage[english]{babel}
\usepackage{algorithm,algpseudocode}

\def\te{\theta}

\def\obs{y}
\def\mis{z}
\def\1{\mathbf{1}}
\def\Id{\mathrm{Id}}
\newcommand{\Cl}{\text{Cl}}
\newcommand{\ka}{\text{ka}}

\newcommand{\Kp}{\mathbf{K_p}}

\newcommand{\Kg}{\mathbf{K_g}}
\def\I0{I_\xi}

\newcommand{\pr}{\mathbb{P}}
\newcommand{\prt}{\mathbb{P}_{\boldsymbol \theta}}
\newcommand{\E}{\mathbb{E}}

\newcommand{\var}{\mathrm{Var}}
 
\def\rme{\mathrm{e}}

\newcommand{\R}{\mathbb{R}}

\newcommand{\btheta}{\boldsymbol{\theta}}

\newcommand{\by}{\mathbf{y}}
\newcommand{\bz}{\mathbf{z}}
\newcommand{\be}{\mathbf{e}}
\newcommand{\bs}{\mathbf{s}}

\def\E{\mathrm{E}}
\def\R{\mathbb{R}}
\def\obs{y}
\def\mis{z}

\def\obs{y}

\newcommand{\Rset}[1]{\mbox{$\mathbb{R}^{#1}$}}
\newcommand{\Nset}{\mbox{$\mathbb{N}$}}

\newcommand{\Sr}{\mbox{$\mathcal{S}$}}

\newcommand{\htheta}{{\widehat{\boldsymbol \theta}}}

\newcommand{\pscal}[2] { \left< #1 , #2 \right> }

\newcommand{\Dt}[1]{\partial_{\boldsymbol \theta} #1}

\newcommand{\Fr}{\mbox{$\mathcal{F}$}}

\journalname{}

\graphicspath{{./../figures/}}
\DeclareGraphicsExtensions{.pdf, .jpg, .jpeg, .png, .gif,.eps}

\begin{document}

\title{Properties of the Stochastic Approximation EM Algorithm with Mini-batch Sampling }



\author{Estelle Kuhn         \and
        Catherine Matias \and
        Tabea Rebafka
}


\institute{Estelle Kuhn \at
              MaIAGE, INRAE, Universit\'e Paris-Saclay\\ Jouy-en-Josas, France \\
              \email{estelle.kuhn@inrae.fr}           
           \and
           Catherine Matias \at
              Sorbonne Universit\'e, Universit\'e de Paris, CNRS\\
 Laboratoire de Probabilit\'es, \\Statistique    et    Mod\'elisation (LPSM)\\  Paris, France\\
               \email{catherine.matias@math.cnrs.fr}           
           \and
           Tabea Rebafka \at
              Sorbonne Universit\'e, Universit\'e de Paris, CNRS\\
 Laboratoire de Probabilit\'es, \\Statistique    et    Mod\'elisation (LPSM)\\  Paris, France\\
               \email{tabea.rebafka@upmc.fr}           
}

\date{Received: date / Accepted: date}

\maketitle

\begin{abstract}
To deal with very large datasets a mini-batch version of the Monte Carlo Markov Chain Stochastic Approximation Expectation-Maximization algorithm for general latent variable models is proposed. For   exponential models 
the algorithm   is shown to be convergent under classical
conditions as the number of iterations 
increases.
Numerical experiments illustrate the performance of the mini-batch algorithm in various models.
In particular, we highlight that   mini-batch sampling results in an important  speed-up of the convergence  of the sequence of estimators generated by the algorithm. 
Moreover, insights on the effect of the mini-batch size on the limit distribution are presented.  
Finally, we illustrate how to use mini-batch sampling in practice to improve results when a constraint on the  computing time is given.
\end{abstract}

\section{Introduction}
On very large datasets the computing time of the classical expectation-maximization (EM) algorithm \citep{DempsterLR}  
as well as its variants such as Monte Carlo EM, Stochastic Approximation EM, Monte Carlo Markov Chain-SAEM and others can be very long, since all data points are visited in every iteration. 
To circumvent this problem, a number of  EM-type algorithms have been proposed, namely various mini-batch
\citep{NealHinton99,LiangKlein2009,Belhal18,Forbes19} and online \citep{titterington84,Lange95,CappeMoulines2009,Cappe2011}
versions of the EM algorithm. They all consist in using only a part of the observations during one iteration in order to shorten computing time and 
accelerate convergence. While online algorithms process a single observation per iteration handled in the order of
arrival,  mini-batch algorithms use larger, randomly chosen subsets of observations. The size of these subsets of data
is generally called the mini-batch size.  
Choosing  large mini-batch sizes entails  long computing times, while very
small mini-batch sizes as well as online algorithms may result in a loss of  accuracy of the algorithm.  {
This raises the question about the optimal mini-batch size
 that  would achieve a compromise between accuracy and computing time. However this issue is generally overlooked.  }

In this article, we propose a mini-batch version of the 
MCMC-SAEM
algorithm \citep{delyon1999,kuhn2004}. 
The original MCMC-SAEM algorithm is a powerful alternative to EM when the E-step is intractable. 	 
	 This is particularly interesting  for nonlinear models or non-Gaussian models, where the unobserved data cannot be simulated exactly from the conditional distribution. 	 Moreover, the
	 MCMC-SAEM algorithm is also more computing efficient than the MCMC-EM algorithm, since only a single instance of the         latent variable is sampled at every iteration of the algorithm. 
Nevertheless, when the dimension of the latent variable is huge, the simulation step  as well as the update of the sufficient statistic can be time consuming. From this point of view the
here proposed mini-batch version is computationally more efficient than the original MCMC-SAEM,  since
at each iteration  only a small proportion of the latent variables is  simulated and only the corresponding data are
visited to update the parameter estimates. 
For 
exponential models, 
we prove  almost-sure convergence of the sequence of estimates generated by the  mini-batch MCMC-SAEM algorithm
under classical conditions as the number of iterations of the algorithm increases.
We also conjecture an asymptotic normality result and the relation between the limiting
  covariance and the mini-batch size. 
Moreover,  for various models we assess  via   numerical experiments   the influence of the mini-batch size on the speed-up of the convergence at the beginning of the algorithm
as well as its impact on the limit distribution of the estimates. Furthermore, we study the  computing time   of the algorithm and address the question of how to use mini-batch sampling in practice to improve results.

\section{Latent variable model   and   algorithm}
This section   introduces  the general latent variable model considered throughout this paper and  the original MCMC-SAEM algorithm. Then  the new  mini-batch version of the MCMC-SAEM algorithm is presented.

\subsection{Model and assumptions}
Consider the common  latent variable model with
incomplete (observed) data $\by$  and    latent (unobserved) variable $\bz$. 
{Denote $n$ the dimension of the   latent variable  $\bz =(z_1\dots,z_n)\in\R^n$.}
 In many   models,  $n$ also  corresponds to the number of observations,
 but  it is not necessary  that $\bz$ and $\by$ have the same size or that each observation $y_i$ depends only on a
 single latent component $z_i$, as it is  for instance the case in the stochastic block model, Section \ref{subsec:sbm}. 

 Denote     $\btheta\in \Theta \subset \R^d$ the model parameter of the joint distribution of the complete data~$(\by,\bz)$.
In what follows, 
omitting all dependencies in the observations $\by$, which are considered as fixed realizations in the analysis, we
assume that the complete-data likelihood function has the following form
\begin{equation*}
f(\bz;\btheta) = \exp\left\{-\psi(\btheta)+\pscal{S(\bz)}{\phi(\btheta)}
\right\} c(\bz),
\end{equation*}
where  $\pscal{\cdot}{\cdot}$ is the scalar product,  $S(\bz)$ denotes a vector of sufficient statistics of the model taking its values in some set $\mathcal S$ and $\psi$ and $\phi$ are   functions on $\btheta$.
The posterior distribution of the latent variables $\bz$ given the observations is denoted by $ \pi(\cdot
; \btheta)$.

\subsection{Description of MCMC-SAEM algorithm}

The original MCMC-SAEM algorithm proposed by \citet{kuhn2004}  is appropriate for models, where the classical EM-algorithm cannot be applied due to difficulties in the E-step. In particular, in those models the conditional expectation $\E_{\btheta_{k-1}}[S(\bz)]$ of the sufficient statistic  under the current parameter value $\btheta_{k-1}$ has no closed-form expression.
In the  MCMC-SAEM algorithm the quantity   $\E_{\btheta_{k-1}}[S(\bz)]$ is thus estimated by a stochastic approximation algorithm. This means that the classical E-step is replaced with a
simulation step using a MCMC procedure combined with a stochastic approximation step.
Here, we focus on a version where the MCMC part is  a Metropolis-Hastings-within-Gibbs algorithm~\citep{Robert_book}. 
More precisely, the $k$-th iteration of  the classical MCMC-SAEM algorithm consists of  the following three steps.

\subsubsection{Simulation step}
A new realization $\mathbf{z}_k$ of the latent variable is sampled   from an ergodic Markov transition kernel
$\Pi(\bz_{k-1}, \cdot |\btheta_{k-1})$, whose  stationary distribution  is     the posterior distribution $\pi(\cdot ; {\btheta_{k-1}})$. 
In practice, this simulation is done by performing one iteration of a  Metropolis-Hastings-within-Gibbs algorithm. 
That is, we consider a collection $(\Pi_i)_{1\le i\le n}$ of symmetric random walk Metropolis kernels defined on $\R^n$, where subscript $i$ indicates that $\Pi_i$ 
 acts only on the $i$-th coordinate, see \citet{Fortetal2003}. These kernels  are applied successively to update the components of $\bz$ one after the other.
More precisely,  
let $(\be_i)_{1\le i\le n}$ be the canonical basis of $\R^n$. Then, 
for each $i\in \{1,\dots,n\}$ starting from the $n$-vector $\bz = (z_1,\dots, z_n)$, the proposal in the direction of $\be_i$ is given by
$\bz + x \be_i$, where $x\in\R$ is sampled from a symmetric increment density $q_i$. This proposal is then accepted with
probability $\min\{1 , \pi(\bz + x \be_i ; {\btheta_{k-1}})/\pi(\bz ; {\btheta_{k-1}})\}.$

\subsubsection{Stochastic approximation step}
The approximation of the sufficient statistic  is updated by
\begin{equation}\label{eq:stochapprox}
\bs_k=(1-\gamma_k)\bs_{k-1} +\gamma_kS(\bz_k),
\end{equation}
where  $(\gamma_k)_{k\geq1}$ is a decreasing sequence of positive step-sizes  such that $\sum_{k=1}^\infty\gamma_k=\infty$ and $\sum_{k=1}^\infty\gamma_k^2<\infty$. 
That is, the current approximation $\bs_k$ of the sufficient statistic  is a weighted mean of its previous value $\bs_{k-1}$ and the value
of the sufficient statistic $S(\bz_k)$ evaluated on the current value of the simulated latent variable $\bz_k$.

\subsubsection{Maximization step}
The model parameter $\btheta$ is updated  by $\btheta_k=\hat \btheta(\bs_k)$
\begin{align*}
&\text{with } \hat \btheta(\bs)= \arg\max_{\btheta\in\btheta} L(\btheta;\bs),\\
&\text{ where }
 L(\btheta;\bs)=-\psi(\btheta)+\pscal{\bs}{ \phi(\btheta)}.
\end{align*}
Depending on the model the maximization problem may   have a closed-form solution or not. 

\subsection{Mini-batch  MCMC-SAEM algorithm}
 \label{sec:algomini}
 When the dimension $n$ of the latent variable $\bz$ is large, the simulation step  can be very time-consuming. Indeed,  simulating {\it all} latent components $z_i$  
at every iteration is costly in  time. Thus, according to the spirit of other mini-batch algorithms, updating only a
part of the latent components may speed up the computing time and also the
convergence of the algorithm. With this idea in mind, denote $0<\alpha\leq 1$  the average proportion  of components of the latent variable $\bz$ that are  updated  during one iteration.  

Furthermore, depending on the model, the evaluation of the sufficient statistic $S(\bz_k)$ on the current latent variable $\bz_k$ can be accelerated  as only a part of the components of $\bz_k$ have changed. This brings a further gain in computing time.

\subsubsection{Mini-batch simulation step}
In the  mini-batch version of the MCMC-SAEM algorithm the simulation step consists of two parts.
First,  select the indices of the components $z_i$ of the latent variable that will be updated. 
That is, we sample the number  $r_k$ of indices  from a binomial distribution
Bin$(n,\alpha)$ and then randomly select  $r_k$ indices among $\{1,\dots,n\}$ without replacement. Denote $\mathcal I_k$ this set of selected indices at iteration $k$. 
Second, instead of sampling all   components $z_{k,i}$, 
we only sample from the Metropolis kernels $\Pi_i$   for $i \in \mathcal I_k$ to update only 
 the components  $z_{k,i}$ with  index $i \in \mathcal I_k$.

 \subsubsection{Stochastic approximation  step}
Again this step consists  in  updating the sufficient statistic $\bs_k$ according to Equation \eqref{eq:stochapprox}. 
However, the naive evaluation of the sufficient statistic $S( \bz_k)$ generally involves all data  $\by$ and  thus is time-consuming on large datasets. Though, in most models it    is computationally much more efficient 
to derive the value of $S( \bz_k)$ from its previous value $S(\bz_{k-1})$ by correcting only for the terms that involve recently updated latent components. In general, this amounts to using only a small part of the data and thus speeds up computing.
An  example is detailed in Section \ref{subsec:sbm}.

\subsubsection{Maximization step}
This  step is  identical to the one in the original   algorithm. It does not depend on the mini-batch proportion $\alpha$ in any way: the formulae for the update of the parameter estimates are identical and the computing time of the M-step is the same for any $\alpha$.

\subsubsection{Initialization}
Initial values $\btheta_0$, $\bs_0$ and $\bz_0$ for the model parameter,  the sufficient statistic  and the latent variable, respectively,     have to be chosen  by the user  or at random.

\bigskip
See Algorithm~\ref{algominiSAEM} for a complete description of the algorithm.

\begin{algorithm}[t]
  \caption{Mini-batch MCMC-SAEM}
  \label{algominiSAEM}
\begin{algorithmic}
 \State {\bfseries Input:} data $\by$.
 \State {\bfseries Initialization:} Choose   initial  values $\btheta_0$,    $\bs_0$, $\bz_0$.
\State Set $k=1$.
  \While{not converged}
  \State  Sample $r_k\sim$ Bin($n,\alpha$).
  \State Sample $r_k$ indices from $\{1,\dots,n\}$, denoted by $\mathcal I_k$.
  \State Set 
  $\bz_k = \bz_{k-1}$
  \For{$i\in\mathcal I_k$}
  \State   Sample $ \bz \sim \Pi_i( \bz_k, \cdot |\btheta_{k-1})$. 
\State Set $\bz_k= \bz$. 
\EndFor
\State $\bs_k=(1-\gamma_k)\bs_{k-1} +\gamma_kS(\bz_k).$
 \State Update  parameter  
$\btheta_k=\hat \btheta(\bs_k). $
 \State Increment $k$.
  \EndWhile
\end{algorithmic}
\end{algorithm}

 \section{Convergence of the algorithm}
In this section we show that, under appropriate assumptions, the mini-batch MCMC-SAEM algorithm converges as the number
of iterations increases. 
Note that we consider convergence of the algorithm for a fixed dataset $\by$ when the number of iterations tends to
infinity, and not statistical convergence where the sample size grows.
Other convergence results for  mini-batch EM and SAEM algorithms appear recently in
 \cite{Forbes19} and  \citet[Chapter 7,][]{Karimi_phD}, respectively.
Basically, the assumptions are classical   conditions   
that ensure the convergence of the original   MCMC-SAEM algorithm. So roughly, our theorem says that if the batch
MCMC-SAEM algorithm converges, so does the mini-batch version for any mini-batch proportion $\alpha$.
Compared to the classical MCMC-SAEM algorithm, the mini-batch version involves an additional stochastic part that comes from
the selection of indices of the latent components that are to be updated. This additional randomness is governed by the value of the mini-batch proportion
$\alpha$. 

We also present arguments to explain the impact of the mini-batch proportion $\alpha$ onto the limit distribution of the sequence generated by the algorithm.

\subsection{Equivalent descriptions}
The above description of the simulation step is convenient for achieving maximal computing efficiency. We now focus on
a mathematically equivalent framework that underlines the fact that the mini-batch MCMC-SAEM algorithm 
  formally belongs to the family of  classical MCMC-SAEM  algorithms.

For two kernels $P_1$ and $P_2$, we denote their composition by 
  \[
P_2 \circ P_1 (\bz, \bz') = \int_{\tilde \bz } P_2 (\tilde \bz, \bz') P_1(\bz,d \tilde \bz).  
\]
With this notation at hand, the Metropolis-Hastings-within-Gibbs uses  the  kernel $\Pi= \Pi_n\circ\dots
\circ\Pi_1$. Now, to describe  the mini-batch simulation step in terms of a Markov kernel, we first introduce kernel $ \Pi_{\alpha,i}$, which is
  a mixture of the original kernel $\Pi_i$ and the identity kernel $\Id$,   defined as
\begin{align*}
 \Pi_{\alpha,i}& (\bz, \bz'|\btheta) =\\
 &   \alpha\Pi_i(\bz , (z_1,\dots, z'_i,\dots, z_n)|\btheta)+(1-\alpha)\Id (\bz,\bz').
 \end{align*}
Hence, the mini-batch simulation step corresponds to generating a latent vector $\bz$ according to 
the Markov kernel $\Pi_\alpha =\Pi_{\alpha,n} \circ \dots \circ \Pi_{\alpha,1}$. Indeed,  $ \Pi_\alpha$ can also be written as
\begin{align*} 
&\Pi_\alpha(\bz,\bz'|\btheta)= \\
&\sum_{k=0}^n \alpha^k (1-\alpha)^{n-k}\sum_{1\le i_1< \dots < i_k \leq n} 
(\Pi_{i_k} \circ \dots \circ \Pi_{i_1}) (\bz, \bz'|\btheta). 
\end{align*}
That means that $\Pi_\alpha$ is a mixture of compositions of the original kernels $\Pi_i$ and the identity kernel. 
In other words, the mini-batch MCMC-SAEM algorithm corresponds to the family of classical MCMC-SAEM algorithms with a particular choice of the transition kernel.
Note that $ \Pi_\alpha$ can also be interpreted as a mixture over different trajectories (the choice of indices $i_1<\dots<i_k$ to be
updated) of Metropolis-Hastings-within-Gibbs kernels acting on a part of the latent vector $\bz$.

We now give a third mathematically equivalent description of the simulation step, which will be appropriate  for the analysis of the  theoretical properties of the algorithm. 
In the $k$-th mini-batch simulation step,  we  sample for every $i\in \{1,\dots, n\}$ a  Bernoulli random variable $U_{k,i}$ with parameter $\alpha$. So $U_{k,i}$  indicates 
 whether the latent variable 
 $\bz_{k-1,i}$ is  updated at iteration $k$ or not. Next, we sample  a realization ${\tilde  \bz}_{k,i}$   from the transition kernel  $\Pi_i$ and set
\begin{equation}\label{eq:update_bik}
\bz_{k,i} = U_{k,i}{\tilde \bz}_{k,i}+(1-U_{k,i}) \bz_{k-1,i}.
\end{equation}
In particular, we see from this formula that,  for $\alpha=1$, the sequence 
$(\bz_k)_{k\ge1}$
generated by the batch algorithm is a Markov chain with transition kernel $\Pi=\Pi_n\circ\dots \circ\Pi_1$, what has  already been mentioned above. 
 
 \cite{Fortetal2003} establish results on the geometric ergodicity of hybrid samplers and in particular for the
random-scan Gibbs sampler. The latter is defined as $ n^{-1}\sum_{i=1}^n \Pi_i$, where each
$\Pi_i$ is a kernel on $\R^n$ acting only on the $i$-th component.  More generally, the random-scan Gibbs
  sampler may be defined as $\sum_{i=1}^n a_i \Pi_i$, where $(a_1,\dots,a_n)$ is a probability distribution. This means
  that at each step of the algorithm, only one component $i$ is drawn from the  probability distribution $(a_1,\dots,a_n)$  and then updated. These probabilities may be chosen uniformly ($a_i=1/n$) or, for example,  can be used to  favor  a component that is  more difficult to explore. 
We generalize the results of  \cite{Fortetal2003}  to a setup, where at each step $k$     kernel $\bar\Pi_\alpha$ is used  that is iterated from  a random-scan Gibbs  sampler $\widetilde \Pi_\alpha$ as follows 
\begin{align}\label{def:noyauPiavecUi}
\bar  \Pi_\alpha (\cdot,\cdot|\btheta_k) &=
\widetilde  \Pi_\alpha (\cdot,\cdot|\btheta_k)^{\sum_i U_{k,i}}, 
\end{align}
with
\begin{align*}
&\widetilde \Pi_\alpha (\cdot,\cdot|\btheta_k) =\\
&\left\{
    \begin{array}{ll}
      \left( \sum_{i=1}^n U_{k,i} \right)^{-1} \sum_{i=1}^n U_{k,i} \Pi_i (\cdot,\cdot|\btheta_k) & \text{ if } \sum_{i=1}^n  U_{k,i} \ge 1, \\
      \Id & \text{ else}                                                                
    \end{array}
    \right. 
\end{align*}
  Note that this is not exactly the kernel corresponding to the algorithm described above, as the same component $i$  can possibly be updated   more than once during the same iteration. Nonetheless, we neglect this effect and
  establish our result for the algorithm based on    kernel~$\bar \Pi_\alpha$.

 \subsection{Assumptions and convergence result}
Assume that the random variables $\bs_0, \bz_1, \bz_2, \dots $ are defined on the same probability space $(\Omega, \mathcal{A}, P)$. We denote $\Fr = \{ \Fr_k \}_{k \geq 0}$ the increasing family of
$\sigma$-algebras generated by the random variables $\bs_0, \bz_1, \bz_2, \dots, \bz_k$. Consider the following regularity conditions.

\begin{description}
\item[({\bf M1})]The parameter space $\btheta$ is an open subset of $\mathbb{R}^{d}$. The  complete-data likelihood
  function is given  by
\begin{equation*}
f(\bz;\btheta)
= \exp\left\{-\psi(\btheta) + \pscal{S(\bz)}{\phi(\btheta)}\right\} c(\bz),
\end{equation*}
where $S$ is a continuous function on $\R^n$   taking its values in an open subset $\Sr$
of $\mathbb{R}^{m}$. Moreover, the convex hull of $S(\mathbb{R}^{n})$  is included in $\Sr$, and for all $\btheta \in \btheta$, 
$$\int |S(\bz)|  \pi(\bz ; \btheta) d\bz < \infty.$$
\item[({\bf M2})] 
The functions $\psi$ and $\phi$ are twice
continuously differentiable on $\btheta$.
\item[({\bf M3})]
The function $\bar{s} : \btheta \rightarrow \Sr$ defined as
$$
\bar{\bs}(\btheta) = \int S(\bz) \pi(\bz; \btheta) d\bz,
$$
is continuously differentiable on $\btheta$.
\item[({\bf M4})] The observed-data log-likelihood function $\ell : \btheta \rightarrow \Rset{}$   defined as  
$$
\ell(\btheta) =  \log \int f(\bz;\btheta) d\bz,
$$
is continuously differentiable on $\btheta$ and
$$
\Dt \int f(\bz; \btheta) d\bz= \int \Dt f(\bz; \btheta)  d\bz.
$$
\item[({\bf M5})] Define $L : \Sr \times \btheta \to \Rset{}$ as $
L(\bs; \btheta)= - \psi(\btheta) + \pscal{\bs}{\phi(\btheta)}.$ 
  There exists a continuously differentiable function
$\htheta : \ \Sr \rightarrow \btheta$, such that, for any $\bs \in \Sr$ and any $\btheta \in \btheta$,
\begin{equation*}
L(\bs; \htheta(\bs))\geq L(\bs; \btheta).
\end{equation*}
\end{description}
 
We now introduce the usual  conditions that ensure convergence of the SAEM procedure. 
\begin{description}
\item[({\bf SAEM1})]
For all $k \in\Nset$, $\gamma_k \in [0,1]$,
$\sum_{k=1}^\infty \gamma_k = \infty$ and  $\sum_{k=1}^\infty \gamma_k^2 <
\infty$.
\item[({\bf SAEM2})]
 The functions  $\ell  : \btheta  \rightarrow  \Rset{}$ and  $\htheta :  \Sr
  \rightarrow \btheta$ are $m$ times differentiable, where we recall that $\Sr$ is an open subset of $\Rset{m}$.
\end{description}
 
 For any $\bs \in \Sr$, we define $H_{\bs}(\bz) = S( \bz)-\bs$ and its expectation with respect to the posterior
 distribution $\pi (\cdot ; \hat \btheta(\bs)) $ denoted by $h(\bs)=\E_{\hat\btheta(\bs)}[S(\bz)  ]-\bs$. For any
 $\rho>0,$ denote $V_\rho(\bz)= \sup_{\btheta\in\btheta}[\pi(\bz ; \btheta)]^\rho$. We consider the following additional
 assumptions as done in \cite{Fortetal2003}.

\begin{description}
\item[({\bf H1})]
There exists a constant $M_0$ such that
\begin{align*}
\mathcal L&=
\left\{\bs\in\mathcal S, \langle \nabla \ell(\hat\btheta(\bs)), h(\bs)\rangle=0
\right\}\\
&\subset
\{\bs\in\mathcal S, -\ell(\hat\btheta(\bs))<M_0
\}.
\end{align*}
In addition, there exist  $M_1\in(M_0,\infty]$ such that
$\{\bs\in\mathcal S, -\ell(\hat\btheta(\bs)) \leq M_1
\}$
is a compact set. 
\item[({\bf H2})]
 The family $\{q_i\}_{1\le i \le n}$ of symmetric densities   is such that, for $i = 1,\dots,n$, there exist   constants $\eta_i > 0$ and $\delta_i <\infty$   such that $q_i (x) > \eta_i$ for all $|x|< \delta_i$.
\item[({\bf H3})] There are constants $\delta$ and $\Delta$ with $0\le\delta\le\Delta\le\infty$ such that 
\begin{align*}
\inf_{i=1,\dots,n}\int_\delta^\Delta q_i(x)dx>0,
\end{align*}
and for any sequence $\{\bz^j\}$ with $\lim_{j\to\infty}|\bz^j|=\infty$,  a subsequence $\{\check \bz^j\}$ can be extracted with the property that, for some $i\in\{1,\dots,n\}$,  for any  $x\in[\delta,\Delta]$ and any $\btheta \in \btheta$,
\begin{align*}
&\lim_{j\to\infty} \frac{\pi(\check \bz^j;\btheta)}{\pi(\check \bz^j-{\rm sign}(z^j_i)x\be_i;\btheta)}=0\quad\text{ and }\\ 
&\lim_{j\to\infty} \frac{\pi(\check \bz^j+{\rm sign}(z^j_i)x\be_i;\btheta)}{\pi(\check \bz^j;\btheta)}=0.
\end{align*}
\item[({\bf H4})] There exist $C>1$, $\rho\in(0,1)$  and $\btheta_0 \in \btheta$ such that, for all $\bz\in\R^n$,
\begin{align*}
|S(\bz)|\le C \pi(\bz;\btheta_0)^{-\rho}.
\end{align*}

\end{description}
To state our convergence result,  we consider the version of the algorithm with truncation on random boundaries studied by
\citet{Andrieu2005}. This additional projection step ensures in particular the stability of the algorithm for the
theoretical analysis and is only a technical tool for the proof without any practical consequences.

 \begin{theorem}\label{thm:convAlgo}
  Assume that the conditions
  {\rm({\bf M1})--({\bf M5})},   {\rm({\bf SAEM1}),}   {\rm({\bf SAEM2})} and   {\rm({\bf H1})--({\bf H4})} hold.
Let  $0<\alpha \leq 1$ and  $(\btheta_k)_{k\ge1}$ be a sequence generated by the mini-batch MCMC-SAEM algorithm with
corresponding Markov kernel  $\bar\Pi_\alpha (\cdot,\cdot|\btheta)$ and truncation on random boundaries
  as in \citet{Andrieu2005}. Then, almost surely,
$$\lim_{k\to\infty}d\left(\btheta_k, \{\btheta, \nabla \ell(\btheta)=0\}\right)=0,$$
where $d(x,A)=\inf\{y\in A,|x-y|\}$, that is, $(\btheta_k)_{k\ge1}$ converges to 
   the set  of critical points of the observed likelihood $\ell(\btheta)$ as the number of iterations increases.
\end{theorem}


\begin{proof}
  The proof consists of  two steps. First, we prove the convergence of the sequence of sufficient statistics $(\bs_k)_{k\ge1}$ towards the set of zeros of  function $h$
  using Theorem 5.5 in \cite{Andrieu2005}. Second, following  the usual reasoning for EM-type
  algorithms, described for instance in \cite{delyon1999}, we deduce that the sequence $(\btheta_k)_{k\ge1}$ converges to the set of  critical
  points of the observed data log-likelihood $\ell$.

  \paragraph*{First step.}~ In order to apply Theorem 5.5 in \cite{Andrieu2005}, we need to establish that their
  conditions (A1) to (A4) are satisfied. In what follows, (A1) to (A4)   refer to the conditions stated in
  \cite{Andrieu2005}. 
 First, note that under our assumptions 
 ({\bf H1}), ({\bf M1})--({\bf M5}) and ({\bf SAEM2}), condition (A1) is satisfied.
Indeed, this is a consequence of Lemma 2 in \citet{delyon1999}.
To establish \rm{(A2)} and \rm{(A3)}, as suggested in \citet{Andrieu2005}, we establish their drift conditions
(DRI), see Proposition 6.1 in \citet{Andrieu2005}. We first focus on establishing (DRI1) in \citet{Andrieu2005}.
To this aim, we rely on \cite{Fortetal2003} that establish results for the random-scan Metropolis sampler. In their context, they
consider a sampler $\Pi = n^{-1}\sum_{i=1}^n \Pi_i$.
We generalize their results to our setup according to \eqref{def:noyauPiavecUi}.
Following the lines of the proof of Theorem 2 in \cite{Fortetal2003}, we can show that Equations (6.1) and (6.3) appearing in the drift condition (DRI1) in \citet{Andrieu2005} are satisfied when  
({\bf H2})--({\bf H3}) hold. Indeed following the strategy developed in  \cite{alla2010}, we first establish Equations (6.1) and (6.3) using a drift function depending on $\btheta$, namely $V_{\btheta}(\bz)=\pi(\bz;\btheta)^{-\rho}$, where $\rho$ is given by ({\bf H4}). Then we define the common drift function $V$ as follows. Let $\btheta_0 \in \btheta$ and $\rho$ be given in ({\bf H4}) and define $V(\bz)=\pi(\bz;\btheta_0)^{-\rho}$. Then for any compact $\mathcal{K}\in \btheta $, there exist two positive constants $c_{\mathcal{K}} $ and $C_{\mathcal{K}}$ such that for all $\btheta \in \mathcal{K}$ and for all $\bz$, we get $c_{\mathcal{K}} V(\bz) \leq \pi(\bz;\btheta)^{-\rho} \le C_{\mathcal{K}} V(\bz)$. We  then establish Equations (6.1) and (6.3) for this drift function $V$. Moreover, using Proposition 1 and Proposition 2 in \cite{Fortetal2003} we obtain that Equation (6.2) in  (DRI1)  from \citet{Andrieu2005} holds. 
Under assumption ({\bf H4}) we have the first part of (DRI2) in \citet{Andrieu2005}. The second part  is true  in our case with $\beta=1$.
Finally,
 (DRI3) in \citet{Andrieu2005} holds in our context with $\beta=1$, since $\bs\mapsto \hat\btheta(\bs)$ is twice continuously differentiable and thus Lipschitz on any compact set. To prove this, we decompose the space in an acceptance region and a rejection region and  consider the integral over four sets leading to   different expressions of the acceptance ratio \citep[see, for example, the proof of  Lemma 4.7 in][]{fortetal2015}.
 This implies that (DRI) and therefore (A2)--(A3) in  \citet{Andrieu2005} are satisfied. 
Notice that  ({\bf SAEM1})   ensures (A4). This concludes the first step of the proof.
\paragraph*{Second step.} ~
   As the function $\bs \mapsto \hat\btheta(\bs)$ is continuous, the second step is immediate by applying Lemma 2 in \cite{delyon1999}.
\end{proof}
 \vfill

\subsection{Limit distribution}\label{subsec:limitvar}

The theoretical study of the impact of the  mini-batch proportion $\alpha$ on the limit distribution of  the sequence $(\btheta_k)_{k\ge1}$ generated by the mini-batch MCMC-SAEM algorithm 
is involved, and here we only present  some heuristic arguments. 
We conjecture that, under reasonable assumptions,   $(\btheta_k)_{k\ge1}$   is asymptotically normal
at rate   $1/\sqrt{k}$ and the limiting covariance matrix, say $V_\alpha$, depends  on the mini-batch proportion $\alpha$ in the following form 
\begin{equation*}
V_\alpha = \frac{2-\alpha}\alpha V_1,
\end{equation*}
where $V_1$ denotes the limiting covariance of the batch algorithm.
This formula is coherent with the expected  behavior with respect to the mini-batch proportion. Namely,  the limit variance is monotone in $\alpha$, for $\alpha=1$ we recover the limit variance of the batch algorithm, and, when $\alpha$ vanishes, the limit variance tends to infinity. 
 Numerical experiments in Sections~\ref{subsec:sbm} and~\ref{subsec:frailty} 
 support this conjecture.

The general approach to establish asymptotic normality of $(\btheta_k)_{k\ge1}$  consists in establishing asymptotic normality of the sequence of sufficient statistics $(\bs_k)_{k\ge1}$ and then applying the delta method. 
Now consider the simple case where the model has a single latent  component, that is, $n=1$.
We rewrite Equation~\eqref{eq:stochapprox} as  
\begin{equation*}
\bs_k=\bs_{k-1}+  \gamma_kh(\bs_k)+  \gamma_k \eta_k,
\end{equation*}
where $ \eta_k= S(\bz_k)-\E_{\btheta_{k-1}}[S(\bz)]$, and note that $S(\bz_{k}) = U_{k}S({\tilde \bz}_{k})+(1-U_{k}) S(\bz_{k-1})$. 
 In general, the principal contribution to  the limit variance comes
from the term
\begin{align} 
\frac1{\sqrt k}\sum_{l=1}^k\eta_l
&=\frac1{\sqrt k}\sum_{l=1}^k  N_{l,k} \left(S(\mathbf{\tilde z}_l)-\E_{\btheta_{l-1}}[S(\mathbf z)]\right) \nonumber\\
&\quad+\frac1{\sqrt k}\sum_{l=1}^k( N_{l,k}-1) \E_{\btheta_{l-1}}[S(\mathbf z)], \label{eq_etak}
\end{align} 
where
\begin{align*} 
 N_{l,k} &=
U_l\sum_{j=l}^k\left\{\prod_{m=1}^{j}(1-U_{l+m})\right\},
\end{align*}
The quantity  $N_{l,k}$ equals zero if there is no update of the unique latent component at iteration $l$. Otherwise, $N_{l,k}$ is the  number of iterations until the next update after the one at iteration $l$.
 The random variable  $N_{l,k}$  takes its values in $\{0,1,\dots, k-l+1\}$ and it can be shown that 
\begin{align*}
\pr(N_{l,k}=m)&=\begin{cases}
1-\alpha,&\quad m=0\\
\alpha^2(1-\alpha)^{m-1},&\quad m=1,\dots,k-l\\
\alpha(1-\alpha)^{k-l},&\quad m=k-l+1
\end{cases}. 
\end{align*} 
Another important property is that 
$\sum_{l=1}^kN_{l,k}=k~a.s.$

It can be shown that the second term in the right-hand side of~\eqref{eq_etak} tends to zero in probability when $k$
goes to infinity. To analyze the first term, we consider the simple setting  where for all
  $k\ge1$, $\btheta_k=\btheta^*$, where $\btheta^*$ is  constant, as e.g. a critical point of the observed likelihood, and  $\tilde\bz_k$  are simulated from the conditional distribution $\pi(\cdot; \btheta^*)$.  In this case, conditionally to the Bernoulli indicators of updates $(U_{k})_{k\ge1}$, the central limit theorem can be applied to the first term. For the  conditional variance of the first term in Equation~\eqref{eq_etak} we obtain 
\begin{align*} 
&\var\left(\frac1{\sqrt k}\sum_{l=1}^k  N_{l,k} \left(S(\mathbf{\tilde z}_l)-\E_{\btheta_{l-1}}[S(\mathbf z)]\right)  \middle| (U_{k})_{k\ge1}\right)\\
&=\frac1{\sqrt k}\sum_{l=1}^k  N_{l,k}^2 V_1,
\end{align*} 
 where $V_1$ is the variance of $S(\mathbf{z})$ with  $\bz\sim\pi(\cdot;
 \btheta^*)$. 
 Further computations yield that the expectation taken over $(U_{k})_{k\ge1}$ gives  
 \begin{align*} 
 \lim_{k\to\infty} \E\left[\frac1{\sqrt k}\sum_{l=1}^k  N_{l,k}^2\right]
 = \frac{2-\alpha}\alpha.
 \end{align*}
 The main difficulty for generalizing this approach to $n>1$ arises from   two facts. First,  the components of $\bz_k$  are not updated  simultaneously, but at different iterations.  Second, the sufficient statistic may not  be linear in $\bz$.  More precisely, when $\bz$ is a vector, $N_{l,k}$ is also a vector and an equivalent of \eqref{eq_etak} is not immediate.

\section{Numerical experiments}

We carry out various numerical experiments  in  a nonlinear mixed effects model,  a Bayesian deformable template model,
the stochastic block model and a frailty model, to illustrate  the performance and the properties of the proposed
mini-batch MCMC-SAEM algorithm and the potential gain in efficiency and accuracy. The programs of these experiments
are available at \\
www.lpsm.paris/pageperso/rebafka/MiniBatchSAEM.tgz

\subsection{Nonlinear mixed model for pharmacokinetic study}
 
Clinical pharmacokinetic studies aim at  analyzing
  the  evolution of the concentration of a given drug in the blood of an individual over a given time interval after absorbing the drug. In this section     a classical one-compartment  model is  considered.  
  
\subsubsection{Model}
The model presented in \citet{Davidian95} serves to  analyze the kinetic of  the drug  theophylline used in therapy for respiratory diseases. For $i=1,\ldots,n$ and  $j=1,\ldots, J$, we define 
\begin{equation*}
\obs_{ij}=h(V_i, \ka_{i},\Cl_{i}) + \varepsilon_{ij}
\end{equation*}
with
\begin{equation*}
h(V_i, \ka_{i},\Cl_{i})=\frac{d_i \ka_{i}}{V_i \ka_{i}-\Cl_{i}}\left[\rme^{-\Cl_i  t_{ij} /V_i} - \rme^{-\ka_{i} t_{ij}}\right] 
\end{equation*}
where the observation $\obs_{ij}$ is the measure of drug concentration on individual~$i$ at time $t_{ij}$.
The drug dose administered to individual~$i$ is denoted $d_i$.
The parameters  for individual~$i$ are
the volume $V_i$ of the central compartment, 
the constant  $\ka_i$ of the drug absorption rate, and
the drug's clearance $\Cl_i$.
The random measurement error is denoted by
$\varepsilon_{ij}$ and supposed to have a centered normal distribution with variance $\sigma^2$. For the individual parameters $V_i $,  $\ka_i$ and $\Cl_i$   log-normal distributions are considered given by
\begin{align*}
 \log V_{i} & =  \log(\mu_V) + \mis_{i,1},\\
  \log \ka_{i} & =  \log(\mu_{\ka}) + \mis_{i,2},\\
  \log \Cl_{i} & =  \log(\mu_{\Cl}) + \mis_{i,3}, 
\end{align*}
where $\boldsymbol z_i=(\mis_{i,1},\mis_{i,2},\mis_{i,3})$ are independent latent random variables following a centered normal distribution with variance
$\Omega = \mathrm{diag}(\omega_{V}^2,\omega_{\ka}^2,\omega_{\Cl}^2)$. 
Then the model parameters are
$\boldsymbol\te=(\mu_V,\mu_{\ka},\mu_{\Cl},\omega_{V}^2,\omega_{\ka}^2,\omega_{\Cl}^2,\sigma^2)$.

\subsubsection{Algorithm} 
We implement the minibatch MCMC-SAEM algorithm presented in Section \ref{sec:algomini}.
In the simulation step we use the following sampling procedure.
 Let $\mathcal I_k$ be the subset of indices of  latent variable components $z_i$ that have  to be updated at iteration $k$. For each  $i\in\mathcal I_k$, we use a Metropolis-Hastings procedure: first, for $l \in \{1,2,3\}$ draw a candidate $\tilde z_{k,i,l}$ from the normal distribution $\mathcal N(z_{k-1,i,l},\eta_l)$, with $\eta_{1}=0.01$, $\eta_{2}=0.02$ and $\eta_{3}=0.03$, chosen to get good mean acceptance rates. Then compute the acceptance ratio  $\rho_{k,i,l}=\rho(z_{k-1,i,l},\tilde z_{k,i,l})$ of the usual Metropolis-Hastings procedure. Finally, draw a realization $\omega_{k,i,l}$  of the   uniform distribution $U[0,1]$ and  set  
$z_{k,i,l}=\tilde z_{k,i,l}$ if $\omega_{k,i,l}<\rho_{k,i,l}$, and $z_{k,i,l}=z_{k-1,i,l}$ otherwise.

In the next  step, we compute the  stochastic approximation  of  sufficient statistics of the model
taking value in $\mathbb{R}^7$ according to 
\begin{equation*}
\boldsymbol s_{k}=(1-\gamma_k)\boldsymbol s_{k-1}+\gamma_k S(\bz_k),
\end{equation*}
where 
\begin{eqnarray*}
S(\bz_k)=\left(\frac1{n}\sum_i  \mis_{i,1}, \frac1{n}\sum_i  \mis_{i,2}, \frac1{n}\sum_i  \mis_{i,3}, \frac1{n}\sum_i  \mis_{i,1}^2, \right.\\
\left.  \frac1{n}\sum_i  \mis_{i,2}^2,  \frac1{n}\sum_i  \mis_{i,3}^2, \frac1{nJ}\sum_{i,j} \left(\obs_{ij}-
h(V_i, \ka_{i},\Cl_{i})  \right)^2 \right)
\end{eqnarray*}
 and where  the sequence $(\gamma_k)_{k\geq1}$ is chosen such that 
 $0< \gamma_k < 1$ for all $k$, $\sum \gamma_k= \infty$ and $\sum \gamma_k^2 < \infty$.

Finally the maximization step  is performed using explicit  solutions  given by the following equations
   \begin{align*}
    \mu_{V,k}&=\exp(s_{k,1}) ; \qquad   &\omega_{V,k}^2= s_{k,4} - s_{k,1}^2 ;\\
\mu_{\ka,k}&=\exp(s_{k,2}) ; \qquad  &\omega_{\ka,k}^2= s_{k,5} - s_{k,2}^2 ;\\
\mu_{\Cl,k}&=\exp(s_{k,3}) ; \qquad &\omega_{\Cl,k}^2= s_{k,6} - s_{k,3}^2;\\
\sigma_k^2&=s_{k,7} .& 
   \end{align*}

 For more technical details on the implementation, we refer to \citep{kuhn2005}.

\subsubsection{Numerical results}
\begin{figure}[t]
\begin{center}
\includegraphics[width=0.5\textwidth]{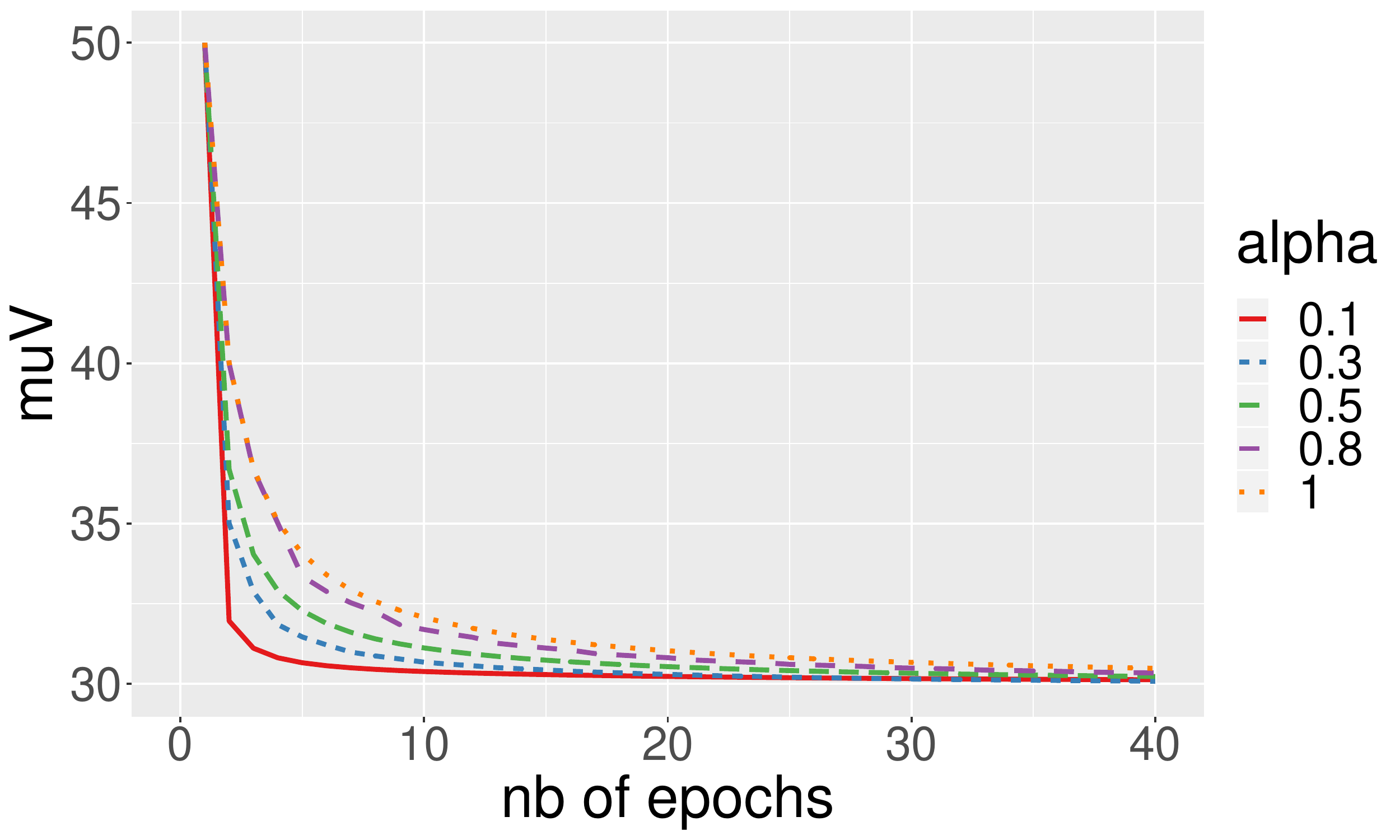} 
\caption{Estimates of the parameter $\mu_V$  using mini-batch MCMC-SAEM  with $\alpha \in \{0.1,0.3,0.5,0.8,1\}$ as a
  function of the number of epochs.}\label{fig:Theomoy}
\end{center}
\end{figure}

\begin{figure}[t]
\begin{center}
\includegraphics[width=0.5\textwidth]{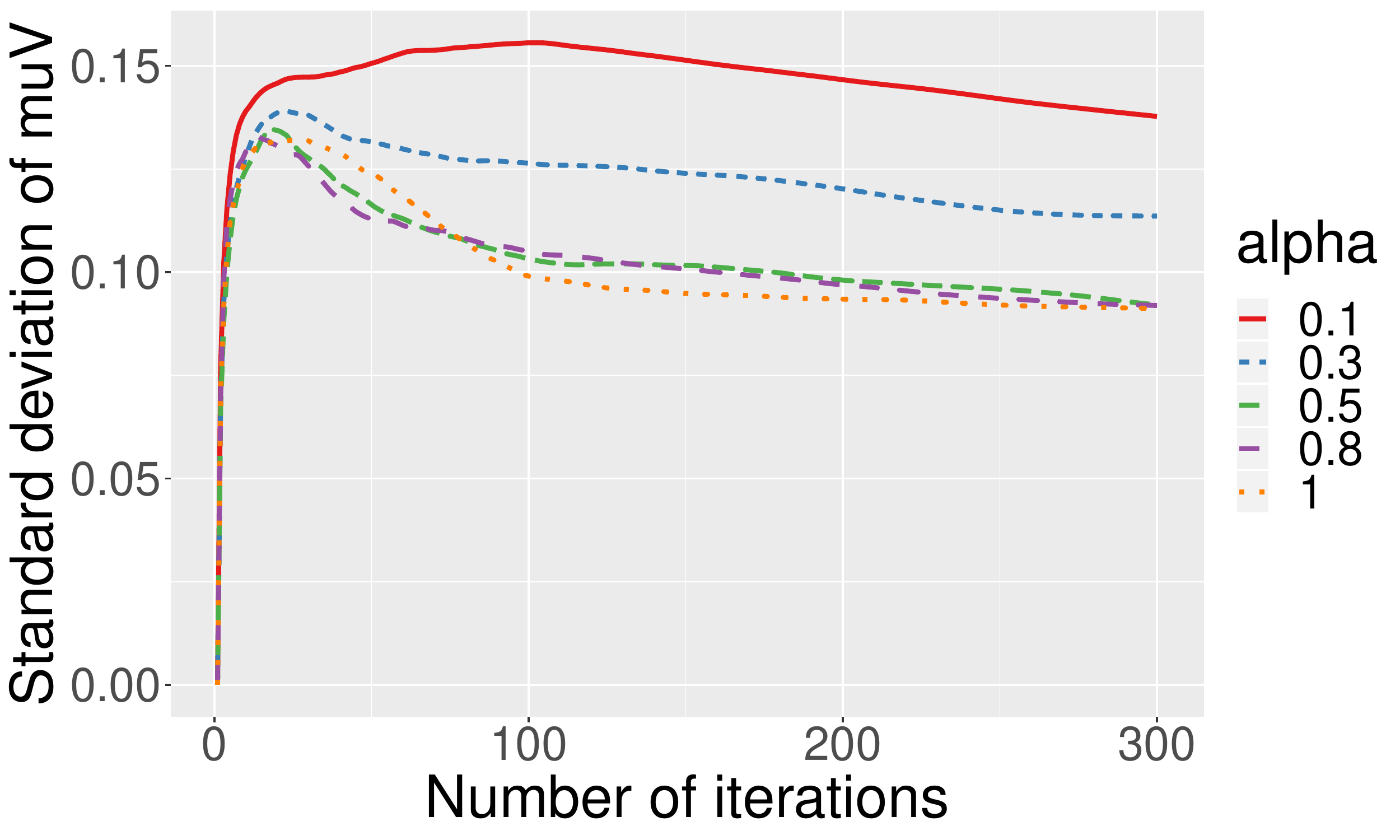} 
\caption{Sample standard deviation of  the estimate of parameter  $\mu_V$  using mini-batch MCMC-SAEM  with  $\alpha \in \{0.1,0.3,0.5,0.8,1\}$ as a function of the number of iterations.}\label{fig:Theosd}
\end{center}
\end{figure}

In a simulation study 
we generate one dataset from the above model with the following parameter values
$n=1000, J=10,  \mu_V=30, \mu_{\ka}=1.8, \mu_{\Cl}=3.5, \omega_{V}=0.02, \omega_{\ka}=0.04, \omega_{\Cl}=0.06, \sigma^2=2$. 
For all individuals $i$ the same dose $d_i=320$ and the    time points  $t_{ij}=j$ are used. 

Then we estimate the model parameters by
both the  original MCMC-SAEM algorithm and our  mini-batch version. The step sizes are  set to
$\gamma_k=1$ for $1 \leq k \leq 50$ and $\gamma_k=(k-50)^{-0.6}$ otherwise. 
This corresponds to a classical choice ensuring Assumption \textbf{(SAEM1)}  \citep{delyon1999}. The
  algorithm did not seem to be sensitive to this choice. 
Several mini-batch proportions, namely  $\alpha \in \{0.1,0.3,0.5,0.8,1\}$, are considered. 
To be more precise,
  the estimation task is repeated $100$ times on the same fixed dataset for all considered algorithms. 
The results for   parameter $\mu_V$ are
shown in Figures \ref{fig:Theomoy} and \ref{fig:Theosd}. The  results for the other parameters are very similar and therefore omitted.

Figure \ref{fig:Theomoy} shows the evolution of the precision of the  running mean of the estimates $\bar\mu_{V,k}=\sum_{l=1}^k\mu_{V,l}/k$ of parameter component  $\mu_V$ as a function
of the number of epochs for different values of the proportion $\alpha$. 
An epoch is   the average number of iterations required to update $n$ latent components.
That is, one epoch corresponds in average to $1/\alpha$ iterations of the mini-batch algorithm with  proportion $\alpha$. 
This means that  parameter estimators are compared when the different algorithms have spent approximately the same time in the simulation step, and, due to the dependency structure of  the nonlinear mixed model, when the algorithms have
 visited  approximately the same amount of data. 
That is, an epoch takes approximately the same computing time for any proportion $\alpha$. So  
 Figure \ref{fig:Theomoy}  compares estimates at comparable computing times. It is obvious that for all algorithms estimation improves when the number of epochs increases. 
Moreover, and more importantly,  the rate of convergence depends  on the mini-batch proportion:  the smaller $\alpha$, the faster the convergence       to the target value $\mu_V=30$.
Here, the fastest convergence is obtained with the smallest mini-batch proportion, that is, $\alpha=0.1$. 
For instance, to attain the precision obtained within 5 epochs with $\alpha=0.1$, we need at least 25 epochs with the batch algorithm $\alpha=1$.

This acceleration of convergence at the beginning of the algorithm induced by mini-batch sampling is characteristic for  mini-batch sampling in any EM-type algorithm. Let us give an intuitive explanation of this characteristic phenomenon. In general, the initial values of the algorithm are far away from the unknown target values. So,  during the first iteration of the batch algorithm, many time-consuming computations are done using a very bad value $\btheta_0$. Only at the very end of the first iteration, the parameter estimate is updated to a slightly better value $\btheta_1$. During the same time, a mini-batch algorithm with small $\alpha$ performs  some computations with the same bad value $\btheta_0$, but after a short time already, the M-step is attained for the first time. As only a couple of latent components have been updated and only a few data points have been visited, the new value $\btheta_1$ may be only a very slight correction of   $\btheta_0$, but, nevertheless,  it is a move into the right direction and the next iteration is performed using a slightly better value than in the previous one. Hence, performing mini-batch sampling consists in making many small  updates of the $\btheta$, while in the same time the batch algorithm only makes very few updates. Metaphorically speaking,  the batch algorithm makes  long and time-consuming steps, but these steps are not necessarily directed into the best direction, whereas  the mini-batch version makes plenty small and quick steps, correcting its direction after every step. As a whole, the mini-batch strategy leads to much better results as illustrated in Figure \ref{fig:Theomoy}.

Figure \ref{fig:Theosd} presents   for
different values of the proportion $\alpha$ the estimates of the empirical  standard deviation with respect to the
number of iterations. We observe that as the number of iterations increases, the standard deviations are lower than for higher values of  $\alpha$. This illustrates in particular that  including more data in the inference task leads to
more accurate  estimation results. This is indeed very intuitive. 
Therefore an  optimal choice of $\alpha$ should
  achieve a trade-off between speeding up the convergence and involving   enough data  in the  process to   get  accurate estimates.

\subsection{Deformable template model for image analysis}\label{subsec:imageanalysis}

\begin{figure}[t]
\begin{center}
\includegraphics[width=0.5\textwidth]{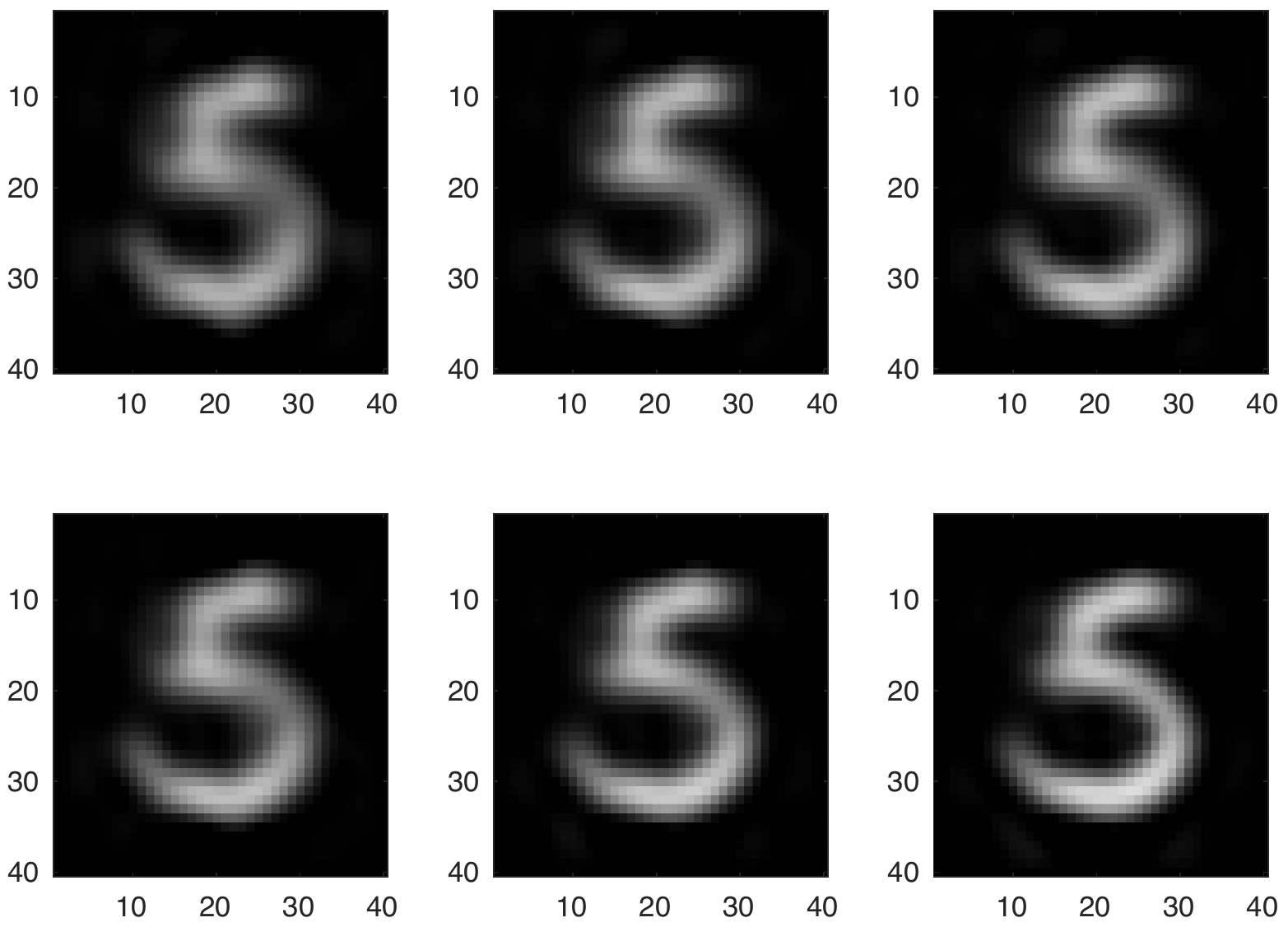} 
\caption{Estimation of the template with $n=20$ images. First row: using batch MCMC-SAEM ; second row: using mini-batch MCMC-SAEM with
  $\alpha=0.1$ ; columns correspond  to 1, 2 and 3  epochs, respectively. }\label{template5}
\end{center}
\end{figure}
\begin{figure*}
\begin{center}
\includegraphics[width=.9\textwidth]{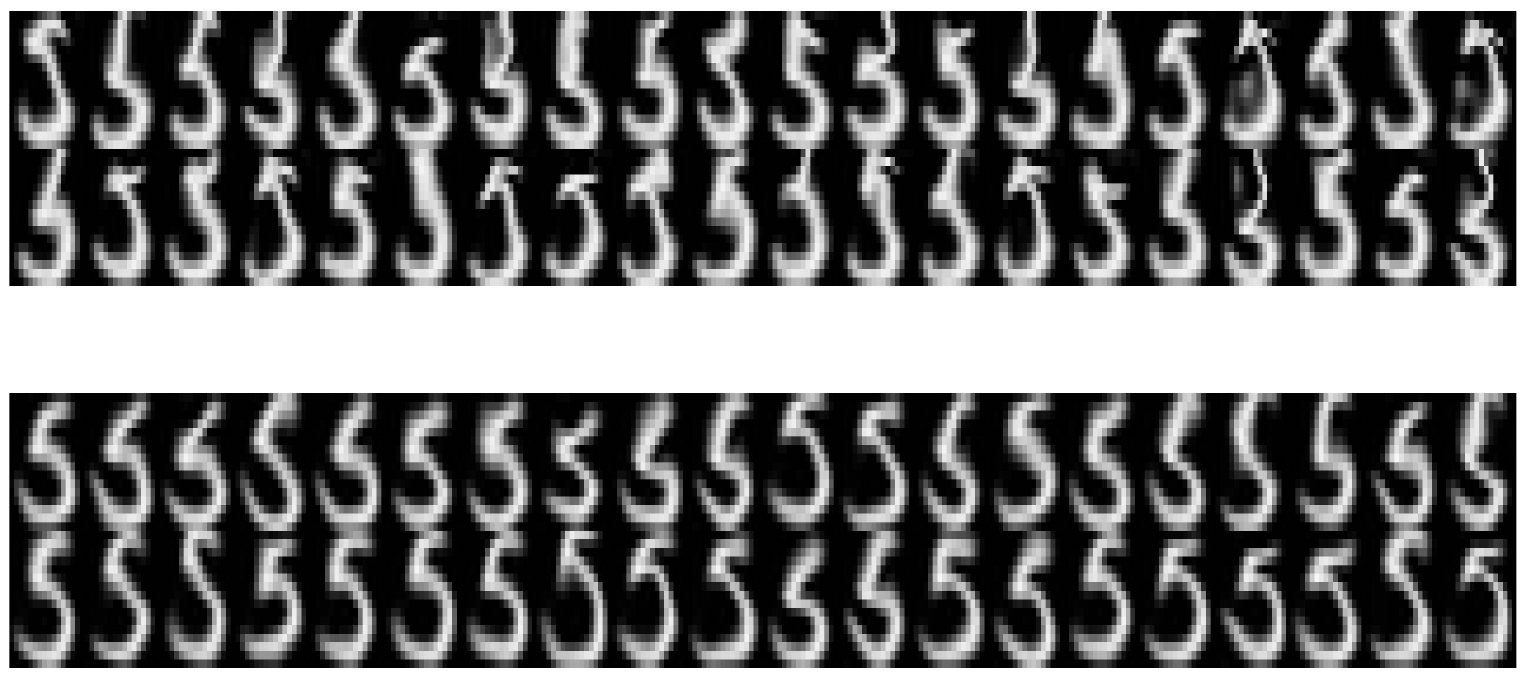} 
\caption{Synthetic images sampled from the model for digit $5$ using the parameter estimates obtained with the batch version on $20$ images  (top) and with the mini-batch version with $\alpha=0.2$ on $100$ images (bottom).}
\label{synthetic5}
\end{center}
\end{figure*}

In this section an example on handwritten digits illustrates the benefits of using mini-batch sampling. 
We consider the dense deformation template model  for image analysis that was first introduced in \citet{alla2007}. This model considers observed images as deformations of a 
common reference image, called template. 

\subsubsection{Model and algorithm}
Using    the formulation    in \citet{alla2010}, let $(y_i)_{1\leq i\leq n}$ be  $n$ observed gray level images. Each image $y_i$ is defined on a grid
of pixels $\Lambda \subset \R^2$, where for each $s\in \Lambda$,
$x_s$ is the location of pixel $s$ in a domain $D \subset
\R^2$. We assume that every image derives from the deformation of a common unknown template $I$, 
 which is a function from $D$ to $\R$. Furthermore,
we assume that for every image $y_i$ there exists 
an unobserved deformation field $\Phi_i:\R^2 \to 
\R^2$  such that
\begin{eqnarray*}
y_i(s) = I(x_s - \Phi_i(x_s)) + \sigma \varepsilon_i (s), 
\end{eqnarray*}
where $ \varepsilon_i(s)$ have standard normal
distribution and  $\sigma^2$ denotes the variance. To formulate this complex problem in a simpler   way,  the
 template $I$ and the deformations $\Phi_i$ are supposed to have  the following parametric form.
Given a set of landmarks denoted by $(p_k)_{1\leq k\leq k_p}$, which covers the
domain $D$, the template function $I$ is parametrized by
coefficients $\xi \in \R^{k_p}$ through
$$
\I0 = \Kp \xi, \quad {\rm where  }\quad  (\Kp \xi)(x) = \sum\limits_{k=1}^{k_p}
\Kp (x,p_k) \xi(k), $$  and $\Kp$ is a fixed   known kernel. 
Likewise, for another  fixed set of landmarks 
$(g_k)_{1\leq k\leq
   k_g} \in D$, 
   the deformation field  is given by 
$$\Phi_i(x) = (\Kg z_i)(x) = \sum\limits_{k=1}^{k_g} \Kg (x,g_k)
(z_i^{(1)}(k),z_i^{(2)}(k)),$$
where $z_i = (z_i^{(1)},z_i^{(2)}) \in \mathbb{R}^{k_g} \times \mathbb{R}^{k_g}$  and again,  $\Kg$ is a fixed   known kernel. The latent variables    $z_i$ are   centered Gaussian variables with covariance matrix~$\Gamma$.
We refer to \citet{alla2010} for further details on the model and also  for the implementation of the MCMC-SAEM algorithm, which
 estimates all model parameters  $(\xi, \Gamma, \sigma^2)$, and so the template $I$.

 \subsubsection{Numerical results}
 In our numerical study we compare the performance of the standard MCMC-SAEM algorithm to  the mini-batch version   on   images from the United States Postal Service database \citep{Hull94}. 
 
In the first experiment we set the mini-batch  proportion $\alpha$ to $0.1$ and have a look on  the   performances
during the very first iterations of the algorithms. 
 Figure \ref{template5} shows the estimated template  for digit $5$ with $n=20$ images after the first three epochs, that is, after three passes through the dataset.
 We observe that   
 the mini-batch algorithm obtains  a more contrasted and accurate template estimate than the batch version.
 That is,   convergence is accelerated   at the beginning of the algorithm when mini-batch sampling is used.
This is very similar to our observations in the nonlinear mixed model in the previous section. 

Now the question is how to take  advantage of this speed up in practice, as we usually do not stop any algorithm after only three epochs. 
As it is good use to run any MCMC-SAEM algorithm until convergence, our approach   is the following. 
Suppose that  we find the computing time acceptable  when running 
the batch algorithm on $n=20$ images   during $1000$ iterations, that is, until convergence.
Can we do something better by using mini-batch sampling, say with $\alpha=0.2$, within the same computing time? When applying the mini-batch algorithm on the same 20 images, convergence is attained much faster and there is a gain in  computing time, but probably also some loss in accuracy. This is not what we are interested in, as we want to make use of the entire allotted computing time. The solution is to increase the number of images in the input of the mini-batch algorithm. Our reasoning is the following: in the batch version 20 images are processed per iteration. So the mini-batch algorithm  with $\alpha=0.2$  can be applied to $n=100$ images, as in average only 20 images are used per iteration.
Hence,  running  the mini-batch algorithm with $\alpha=0.2$ and $n=100$  and the batch version with $n=20$ over the same number of  iterations takes almost exactly the same time. 
We see that, given a constraint on the computing time,
using a small mini-batch proportion allows   to increase the number of images in the input. 

To assess the accuracy of the  estimates obtained by the two algorithms, 
we generate new samples from the model    using the parameter estimates obtained by the different algorithms.
From the synthetic images   presented in  Figure \ref{synthetic5} we see   that the ones in the lower part of the figure
 resemble 
 more  usual handwritten digits $5$   than the ones in the upper part. This  highlights that both template and deformation are better estimated by the mini-batch version performed on $100$ images of the dataset than with the batch algorithm on 20 images. Hence,   given a constraint on the computing time, more accuracy can be obtained by using the mini-batch MCMC-SAEM instead of the original  algorithm.

 \subsection{Stochastic block model }\label{subsec:sbm}
Now we turn to a random graph model  which is interesting as it has  a  complex dependency structure.
The model is  the so-called  stochastic block model  (see \cite{matias14} for a review), where every observation depends on more than one latent component.

 \subsubsection{Model }
 In the stochastic block model (SBM) the latent variable $\mathbf z=(z_1,\dots,z_n)$ is composed of  i.i.d. random variables $z_i$ taking their values in
 $\{1,\dots,Q\}$ with probabilities $p_q=\pr_{\btheta}(z_1=q)$ for $q=1,\dots,Q$.  
 The observed adjacency matrix $\mathbf
 y=(y_{i,j})_{i,j}$ of a directed graph is such that the observations $y_{i,j}$ are independent conditional on
 $\mathbf z$ and $y_{i,j}|\mathbf z$ has Bernoulli distribution with parameter $\nu_{z_i,z_j}$ depending on the latent
 variables of the interacting nodes $i$ and $j$.  
 
 We see that  every observation $y_{i,j}$ depends on two latent components, namely on $z_i$ and $z_j$. In turn,
 the   latent component $z_{k}$  influences all    observations in the set $\{y_{i,k}, i=1,\dots, n\}\cup\{y_{k,j}, j=1,\dots,n\}$   creating complex stochastic dependencies between the observations. 

 Denote $\btheta=((p_{q})_{1\le q\le Q},(\nu_{q,l})_{1\le q,l\le Q})$ the collection of model parameters for a directed SBM.  The complete log-likelihood function is given by 
 \begin{align*}
 &\log \pr_{\btheta}(\mathbf y, \mathbf z)
 = \log \pr_{\btheta}(\mathbf y|\mathbf z)+\log  \pr_{\btheta}(\ \mathbf z)\\
 &=\sum_{q=1}^Q\sum_{i=1}^n\1\{z_{i}=q\}\log p_{q} \\
 &\quad +\sum_{q,l}\sum_{i, j}\1\{z_{i}=q,z_{j}=l\}y_{i,j} \log\nu_{q,l} \\
 &\quad+\sum_{q,l}\sum_{i, j}\1\{z_{i}=q,z_{j}=l\}(1-y_{i,j})  \log
 (1- \nu_{q,l})
 \end{align*}
where $\1\{A\}$ denotes the indicator function of the set $A$. Hence, the complete log-likelihood of the SBM belongs to the exponential family with  the following  sufficient statistics, for $1\leq q,l\leq Q$,
 \begin{align*}
 S_1^q(\mathbf z)
 &=\sum_{i=1}^n\1\{z_{i}=q\},\\
 S_2^{q,l}(\mathbf z)
 &=\sum_{i,j} \1\{z_{i}=q,z_{j}=l\} y_{i,j},\\
S_3^{q,l}(\mathbf z)
 &=\sum_{i,j}\1\{z_{i}=q,z_{j}=l\}(1-y_{i,j}).
 \end{align*}
and  corresponding natural parameters
 $\varphi_1^q(\btheta)=\log p_q$,
 $\varphi_2^{q,l}(\btheta)=\log\nu_{q,l}$ and
 $\varphi_3^{q,l}(\btheta)=\log(1-\nu_{q,l})$.

 \subsubsection{Algorithm}
  \begin{figure}[t]
\begin{center}
\includegraphics[width=0.5\textwidth]{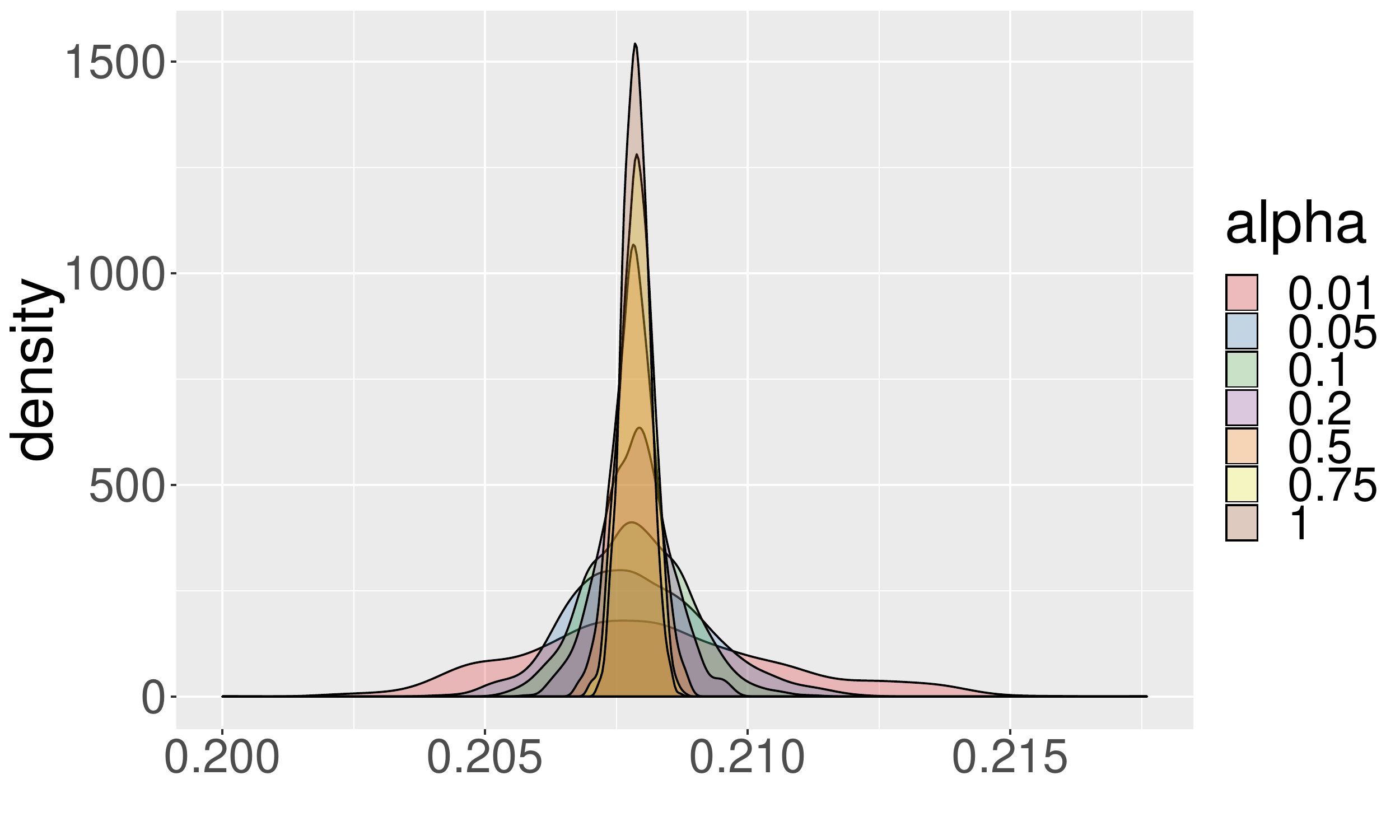} 
\caption{Estimation of the limit distribution of the estimate of $\nu_{2,2}=0.2$ after $10\ 000$   iterations for the mini-batch algorithm with   proportions $\alpha\in\{0.01, 0.05, 0.1, 0.2, 0.5, 0.75, 1\}$. }\label{figSBMhistGamma4}
\end{center}
\end{figure}

 The implementation is straightforward. In the simulation step, 
   the  proposal distribution $q$ of latent variables  in the Metropolis algorithm is the discrete uniform distribution on $\{1,\dots,Q\}$. 
 
As in  other models, the update of the sufficient statistic $S(\mathbf{z})$ can be numerically optimized. 
Indeed, with $\mathcal I_k$ denoting the indices of latent components that are simulated, the vectors $\mathbf z_k$ and $\mathbf z_{k-1}$ are only different in the  components with indices belonging to $\mathcal I_k$. As a consequence, the statistic
$S_1^q(\mathbf z_k)$, for instance, is more rapidly updated by computing
$$S_1^q(\mathbf z_{k-1}) + \sum_{i\in\mathcal{I}_k}(\1\{z_{k,i}=q\} - \1\{z_{k-1,i}=q\})$$
than by the formula $\sum_{i=1}^n\1\{z_{k,i}=q\}$, which involves more operations.
 Likewise,  $S_2^{q,l}(\mathbf z_k)$ is faster computed as follows
\begin{align*} &S_2^{q,l}(\mathbf z_k)
 =S_2^{q,l}(\mathbf z_{k-1}) 
  +\sum_{i\in\mathcal I_k}\sum_{j=1}^n\1\{z_{k,i}=q,z_{k,j}=l\}y_{i,j}\\
  &\quad-\sum_{i\in\mathcal I_k}\sum_{j=1}^n\1\{z_{k-1,i}=q,z_{k-1,j}=l\}y_{i,j}\\
  &\quad+\sum_{i=1}^n\sum_{j\in\mathcal I_k}\1\{z_{k,i}=q,z_{k,j}=l\}y_{i,j}\\
 &\quad- \sum_{i=1}^n\sum_{j\in\mathcal I_k}\1\{z_{k-1,i}=q,z_{k-1,j}=l\}y_{i,j}.
\end{align*}
Here we see that not the entire data $\by$ are visited to update $S_2^{q,l}$, but only those observations $y_{i,j}$ that stochastically depend on the updated latent components $z_i$ with $i\in\mathcal I_k$.

 The maximisation of the complete likelihood function with given values for the sufficient statistics $s^q_1, s^{q,l}_{2}$ and $s^{q,l}_{3}$ 
 is straightforward. The updated parameter values are  given by  $p_{q} ={s^q_{1}}/n$ and $\nu_{q,l}= {s^{q,l}_{2}}/{(s^{q,l}_{2} +s^{q,l}_{3}})$. 

 \subsubsection{Simulation results}

 \begin{figure}[t]
\begin{center}
\includegraphics[width=0.5\textwidth]{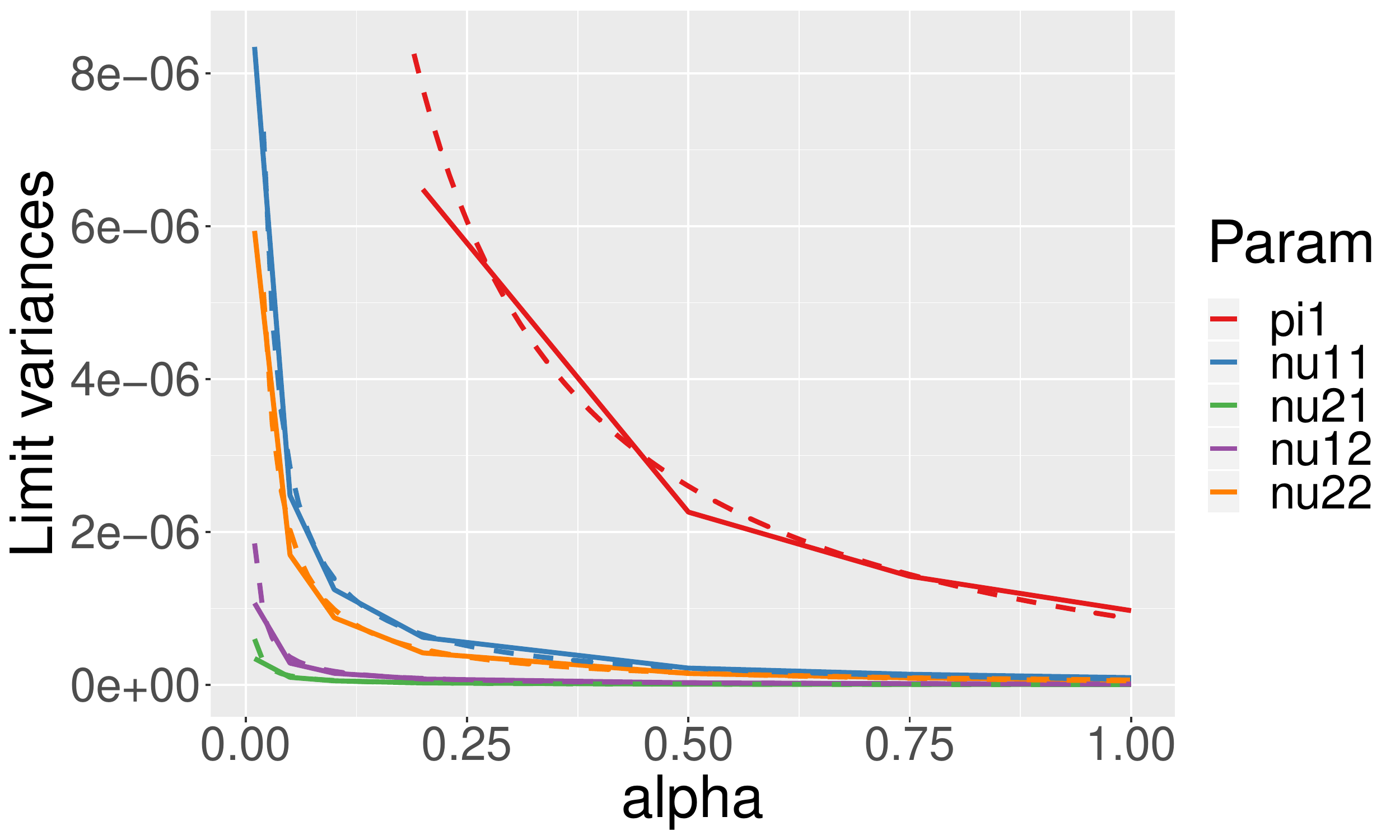} 
\caption{Sample variances of the parameter estimates  after $10\ 000$   iterations as a function of the   mini-batch    proportion $\alpha$ (solid lines) and adjusted theoretical limit variances (dashed lines).  }\label{figSBMlimvar}
\end{center}
\end{figure}

For our simulations,  model parameters are set to $\pi_1=1-\pi_2=0.6$  and 
$$
(\nu_{q,l})_{q,l}=\begin{pmatrix} 0.25 & 0.1 \\ 0.1 & 0.2
\end{pmatrix}.
$$
To study   the impact of   mini-batch sampling  on the asymptotic behavior of the estimates, 
a directed graph  with  $n=100$  nodes   is generated from the model and the mini-batch algorithm with $\alpha$ ranging from $0.01$ to  $1$  and with $10\ 000$ iterations is applied 1000 times. 

Figure \ref{figSBMhistGamma4} shows the histograms of the estimates of parameter $\nu_{2,2}=0.2$ obtained with the different algorithms. All histograms are  approximately unimodal, centered at the same value and symmetric. Moreover, we see that the larger the mini-batch size the tighter the distribution. Indeed, it seems  that the estimates are asymptotically normal and   the limit variance  increases when $\alpha$ decreases. 
This increase of the  limit variance induced by mini-batch sampling  is illustrated for all model parameter estimates in Figure   \ref{figSBMlimvar}. Furthermore,  Figure   \ref{figSBMlimvar} checks  whether the theoretical formula of the limit variance derived in Section \ref{subsec:limitvar} is adequate. Recall that we conjecture that
 the limit variance obtained with the mini-batch algorithm with proportion $\alpha$ equals $(2-\alpha)/\alpha$ times the limit variance obtained with the batch algorithm, here represented by the dashed lines.
The excellent fit supports our conjecture.
 
 \begin{figure}[t]
\begin{center}
\includegraphics[width=0.5\textwidth]{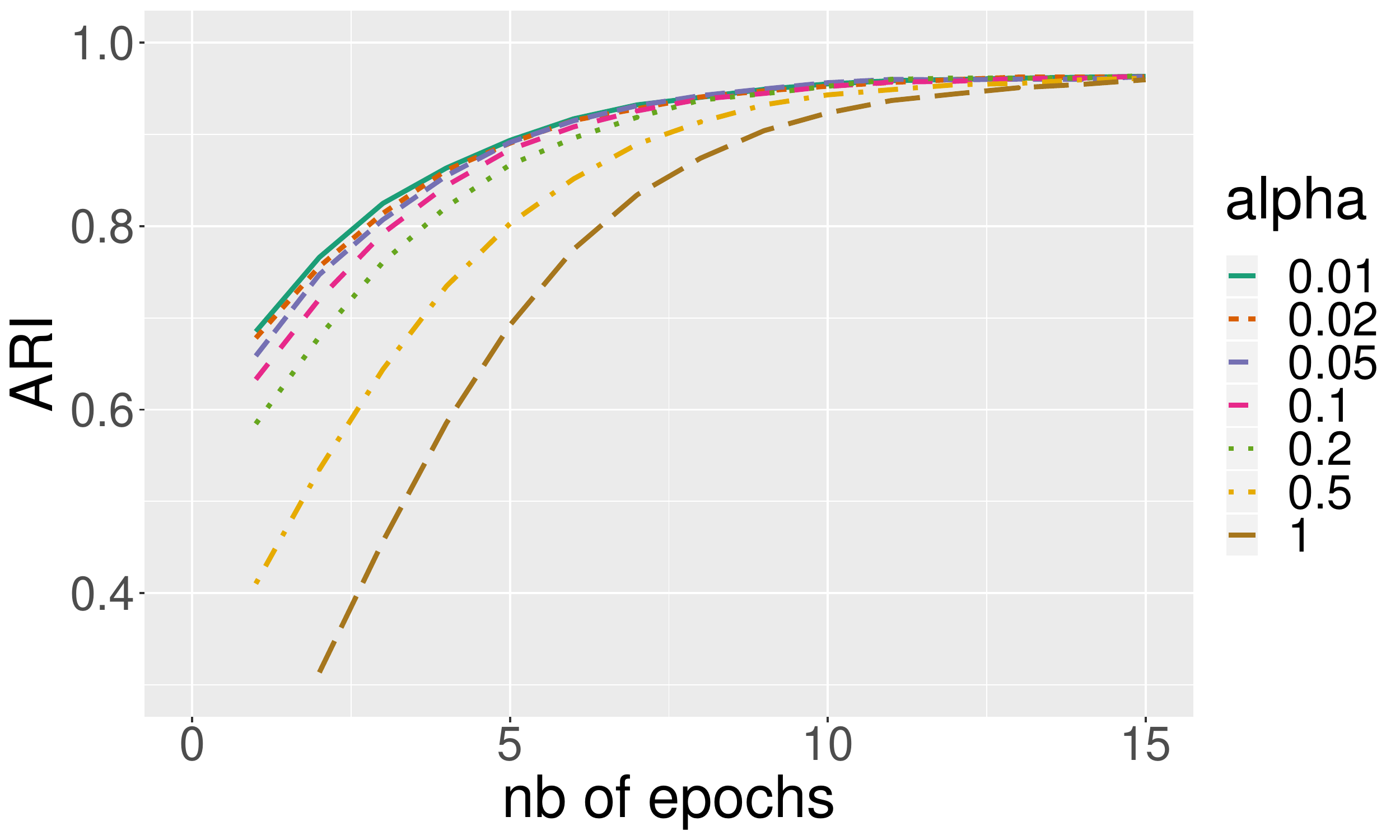} 
\caption{Mean ARI   obtained by mini-batch MCMC-SAEM  algorithms  as a  function of the number of epochs for $\alpha\in\{0.01, 0.02, 0.05, 0.1, 0.2, 0.5, 1\}$. }\label{figSBMaccel}
\end{center}
\end{figure}

Concerning the behavior at the beginning of the algorithm, we observe the same   acceleration for the mini-batch versions as in the other models.  Here 500 datasets are simulated from the SBM   and Figure~\ref{figSBMaccel} shows the evolution of the adjusted rand index (ARI) during the first epochs \citep{HA1985}. This index compares the clustering of the nodes obtained by the algorithm to the true block memberships given by the latent components $z_i$. The ARI equals one if two clusterings are identical up to permutation of the block labels. We see that the algorithm with the smallest mini-batch proportion provides the best clustering at any given number of epochs. Finding the good clustering is essential for accurate parameter estimates.

Finally, we  have a closer look on  computing time aspects. 
As already mentioned, the computing time  of  the M-step does not depend on the mini-batch proportion $\alpha$.
 It is also  clear that the  simulation step in the mini-batch algorithm is in average $\alpha$ times the computing time of the simulation step in the batch version, as only a proportion $\alpha$  of the latent components are simulated. 
However,  the computing time of the stochastic approximation step, and in particular the
 update of the sufficient statistic,   depends on the amount of data used and thus on the dependency structure of the model. In most models, where every observation only depends  on a single latent component $z_i$, the proportion of data involved in the update is $\alpha$. However, in the SBM 
a set of latent components $\{z_1, \dots, z_m\}$ for   $m<n$ influences  the set of observations $\{y_{i,j}, i=1,\dots, m, j=1,\dots,n\}\cup\{y_{i,j}, i=1,\dots,n,j=1,\dots,m\}$, whose cardinality is $2mn -m^2$. It follows that  for a mini-batch proportion $\alpha$ the corresponding proportion of data used 
to update the sufficient statistics is  
$\alpha(2  -\alpha)$ and so the computing time of this step is $\alpha(2  -\alpha)$ times the computing time of the stochastic approximation step in the batch algorithm.

Let us call SAE-step   the combination of the simulation step and the stochastic approximation step.
We determine the median computing time of the SAE-step  over 100 runs of the SAE-step for different mini-batch proportions $\alpha$ and different numbers of nodes $n$.
Figure~\ref{figSBMcompTimeNbnodes} shows the 
ratio of the median   time of the mini-batch SAE-step with proportion $\alpha$  over the median time of the batch SAE-step for different numbers of nodes, that is,
$$
\alpha\mapsto\frac{\text{median   time of SAE-step with $\alpha$}}{\text{median   time of batch SAE-step}}.
$$ 
  The dashed lines represent the functions $\alpha\mapsto\alpha$ and $\alpha\mapsto\alpha(2  -\alpha)$. The first corresponds to the   expected  ratio of computing times of the simulation step, the latter to the  expected  ratio of computing times
of the stochastic approximation step.
As expected, the observed computing time ratios  of the entire SAE-step fall between these two boundaries.

  \begin{figure}[t]
\begin{center}
\includegraphics[width=0.5\textwidth]{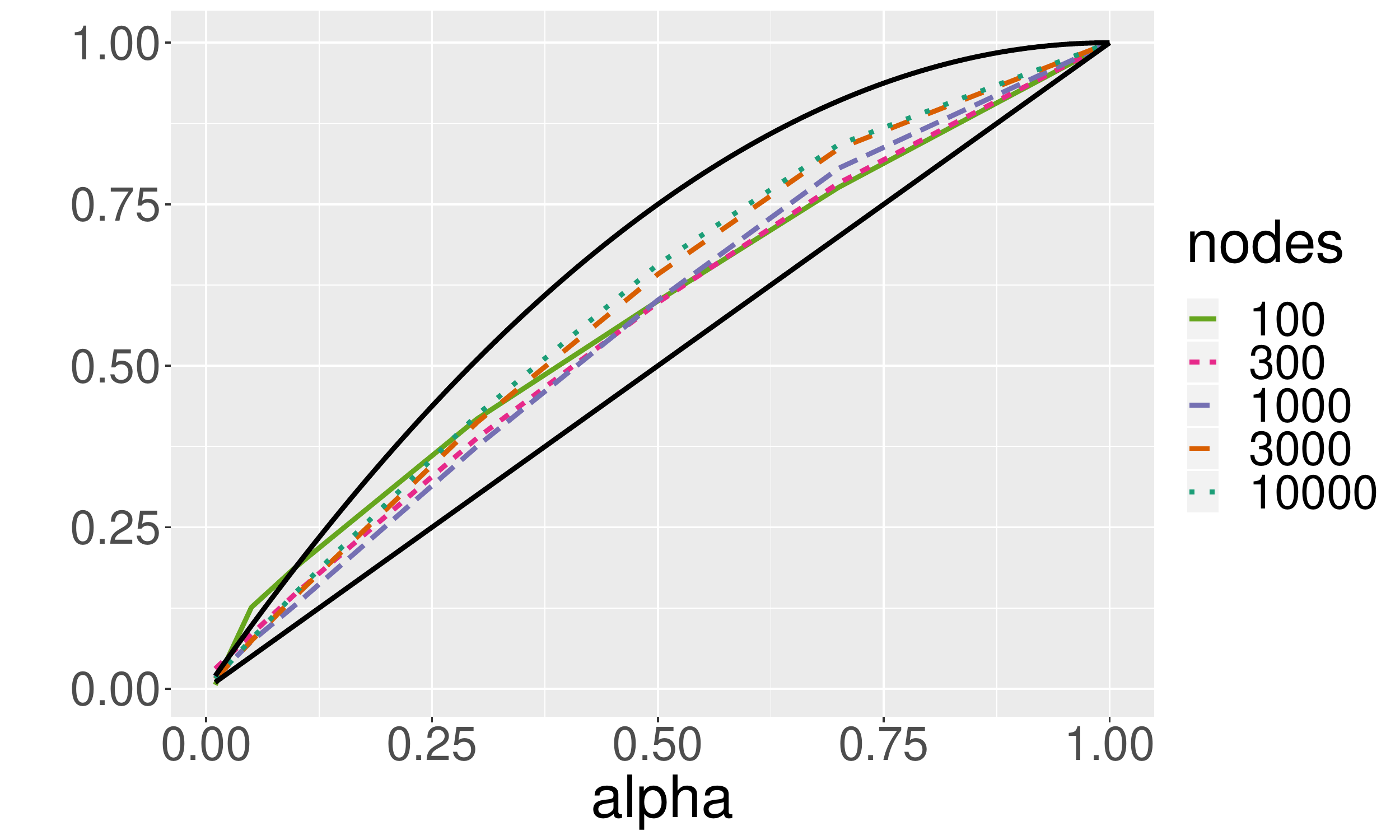}
\caption{
Ratio of the median computing time of the mini-batch SAE-step with $\alpha$ over the median time of the batch SAE-step for different numbers of nodes (solide lines). 
  The dashed lines represent the functions $\alpha\mapsto\alpha$ and $\alpha\mapsto\alpha(2  -\alpha)$, the first corresponds to the   expected  ratio of computing times of the simulation step, the latter to the  expected  ratio of 
of the stochastic approximation step. 
}\label{figSBMcompTimeNbnodes}
\end{center}
\end{figure}

\subsection{Frailty model in survival analysis} \label{subsec:frailty}
In survival analysis the frailty model is an extension of the well-known Cox model \citep{duchateau_book}. The hazard rate function  in the frailty model  includes an additional random effect, called frailty, to account for  unexplained heterogeneity.

\subsubsection{Model and algorithm}

The observations are survival times $\mathbf t= (t_{ij})_{1\le i \le n, 1\le j \le m}$ measured over $n$ groups with $m$  measurements per group, and  covariates ${\mathbf X_{ij}}  \in \R^p$. 
We denote by
 $\lambda_0$ the baseline hazard function, that is  here chosen  to be   the  Weibull function    given by
\[
\lambda_0(t) = \lambda_0 \rho t ^{\rho-1} , \quad  t>0, 
\]
with $\lambda_0 >0$ and $\rho>1$. 
For every group $i$, a latent variable $z_i$ is introduced representing   a frailty term.
We suppose that $z_1,\dots,z_n$ are i.i.d. with centered Gaussian distribution with variance $\sigma^2$.  
 The conditional hazard rate $ \lambda_{ij} (\cdot | z_i) $ of 
observation $t_{ij}$ given the frailty $z_i$  is defined as 
\begin{equation*}
 \lambda_{ij} (\cdot | z_i) = \lambda_0(\cdot) \exp({\mathbf X_{ij}} ^\intercal {\boldsymbol\beta} +z_i),
\end{equation*}
where ${\boldsymbol\beta} \in \R^p$. 
Thus, the unknown model parameter is $\btheta=({\boldsymbol\beta},\sigma^2, \lambda_0, \rho)$. In practical
applications the main interest  lies in the estimation of the regression parameter~${\boldsymbol\beta}$.

In the frailty model the   conditional survival function is given by
\begin{align*}
  G_{ij }(t| z_i) &= \pr (t_{ij} >t |z_i)\\
  &= \exp [- \lambda_0 t^\rho \exp({\mathbf X_{ij}}^\intercal {\boldsymbol\beta} + z_i)]. 
\end{align*}
In other words, the conditional distribution of the survival time $t_{ij}$ given $z_i$ is the Weibull distribution with scale parameter $\lambda_0 \exp({\mathbf X_{ij}}^\intercal {\boldsymbol\beta} + z_i)$
 and shape parameter  $\rho$. 
For the   conditional and the complete likelihood functions we obtain 
\begin{align*}
  \prt (\mathbf t  | \bz) &= \prod_{i=1}^n \prod_{j=1}^m G_{ij} (t_{ij} |
  z_i) \lambda_{ij}(t_{ij} | z_i) , \\
   \prt (\mathbf t, \bz) &= \prt (\mathbf t  | \bz)\prod_{i=1}^n \varphi\left(\frac{z_i}\sigma\right), 
\end{align*}
where $\varphi$ denotes  the density of the standard normal distribution.

 The implementation of the algorithm is straightforward.
In the simulation step candidates $\tilde z_{k,i}$ are drawn  from the normal distribution $\mathcal N(z_{k-1,i},0.2)$.  In the M-step, the updates of $\sigma^2$ and $\lambda_{0,k}$ are explicit, while those of 
 ${\boldsymbol\beta}$ and $\rho$ are  obtained by  the Newton-Raphson  method.

\subsubsection{Numerical results}

\begin{figure}[t]
\begin{center}
\includegraphics[width=0.5\textwidth]{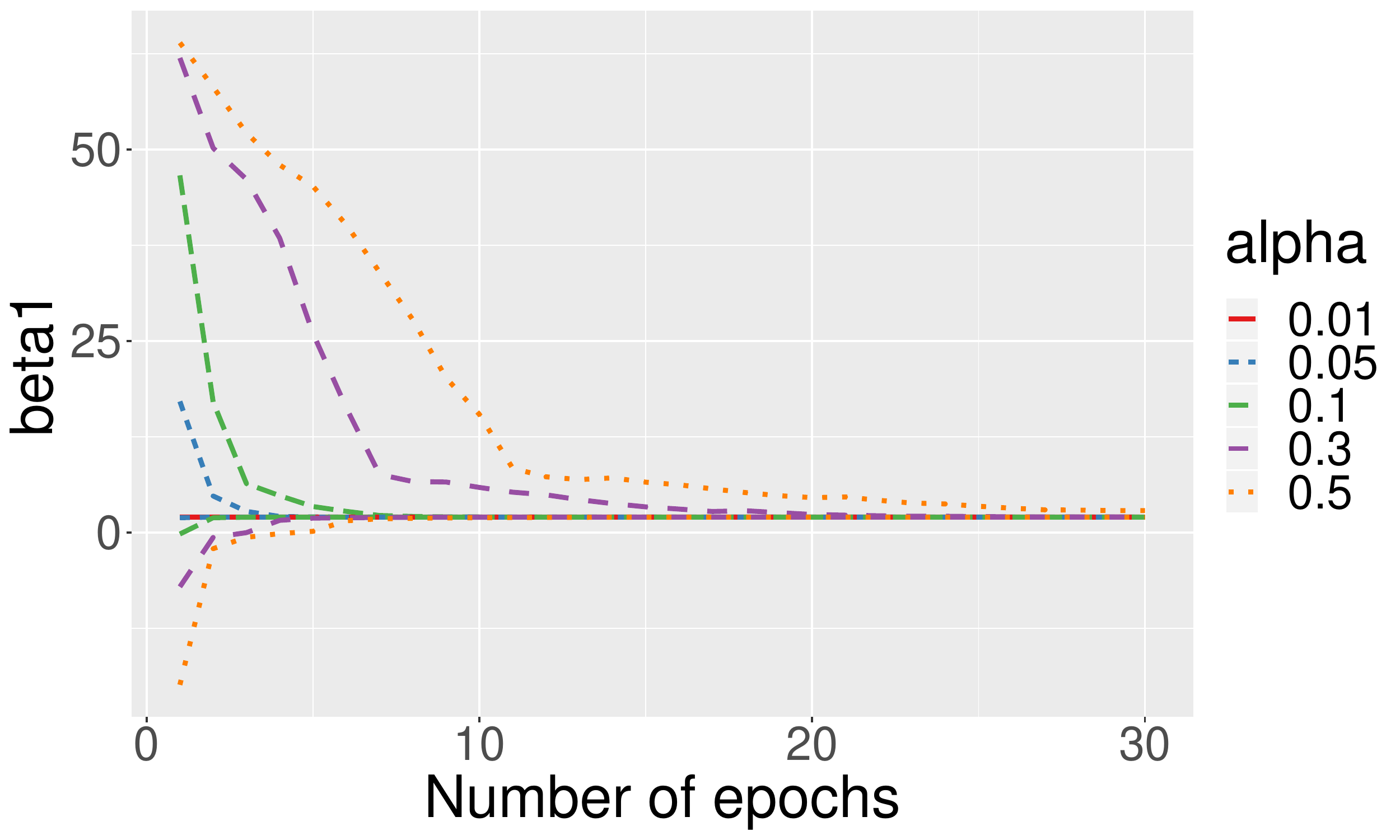}
\caption{Evolution of  80\%-confidence bands for $\beta_1$  with respect to the number of epochs for  mini-batch proportions $\alpha\in\{0.01, 0.05, 0.1, 0.3, 0.5, 0.8, 1\}$}\label{fig_frailty_confint}
\end{center}
\end{figure}

\begin{figure}[t]
\begin{center}
\includegraphics[width=0.5\textwidth]{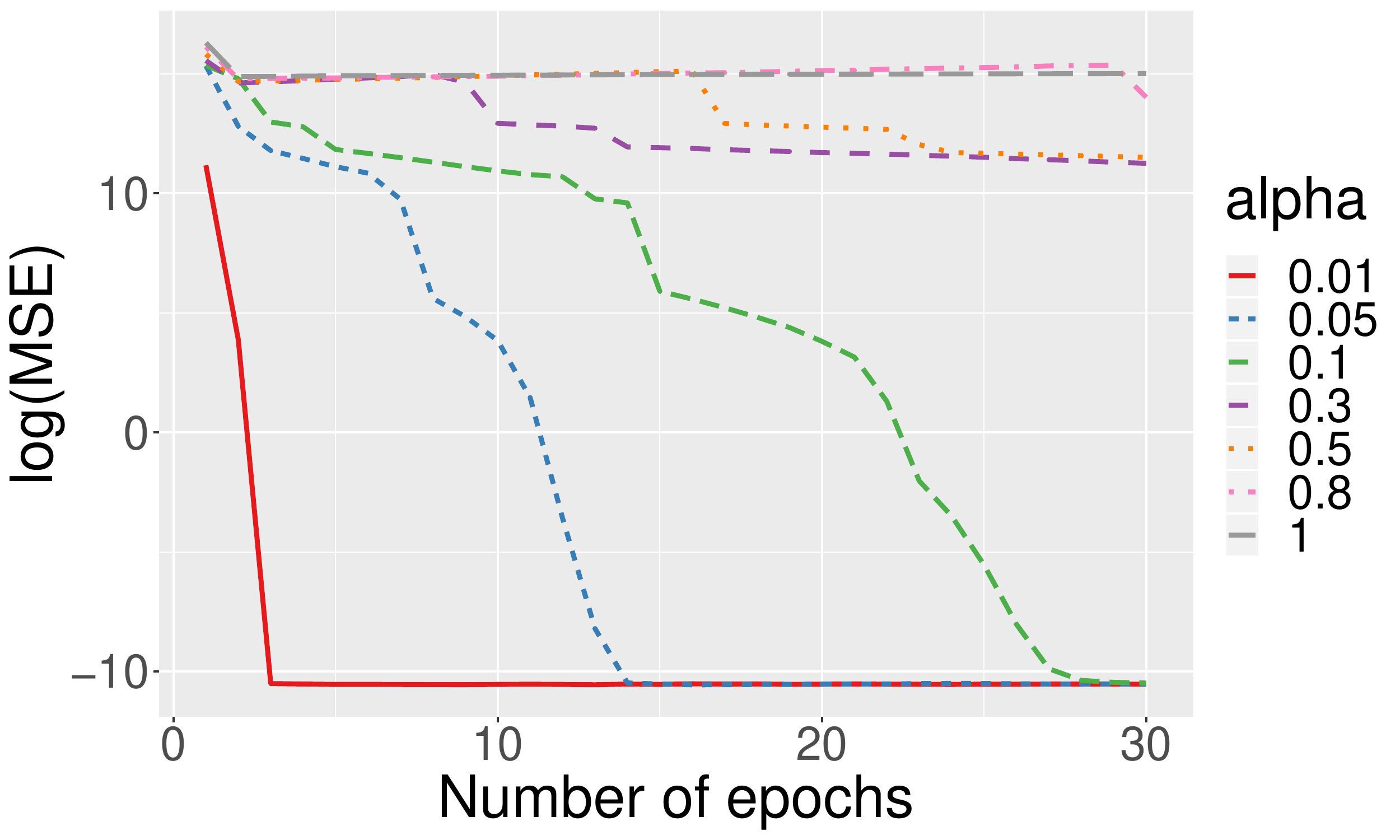}
\caption{Evolution of  the  logarithm of the empirical mean squared error of  estimates of $\beta_1$ with respect to the number of epochs for  mini-batch proportions $\alpha\in\{0.01, 0.05, 0.1, 0.3, 0.5, 0.8, 1\}$.   }\label{fig_frailty_loglik}
\end{center}
\end{figure}

\begin{figure}[t]
\begin{center}
\includegraphics[width=0.5\textwidth]{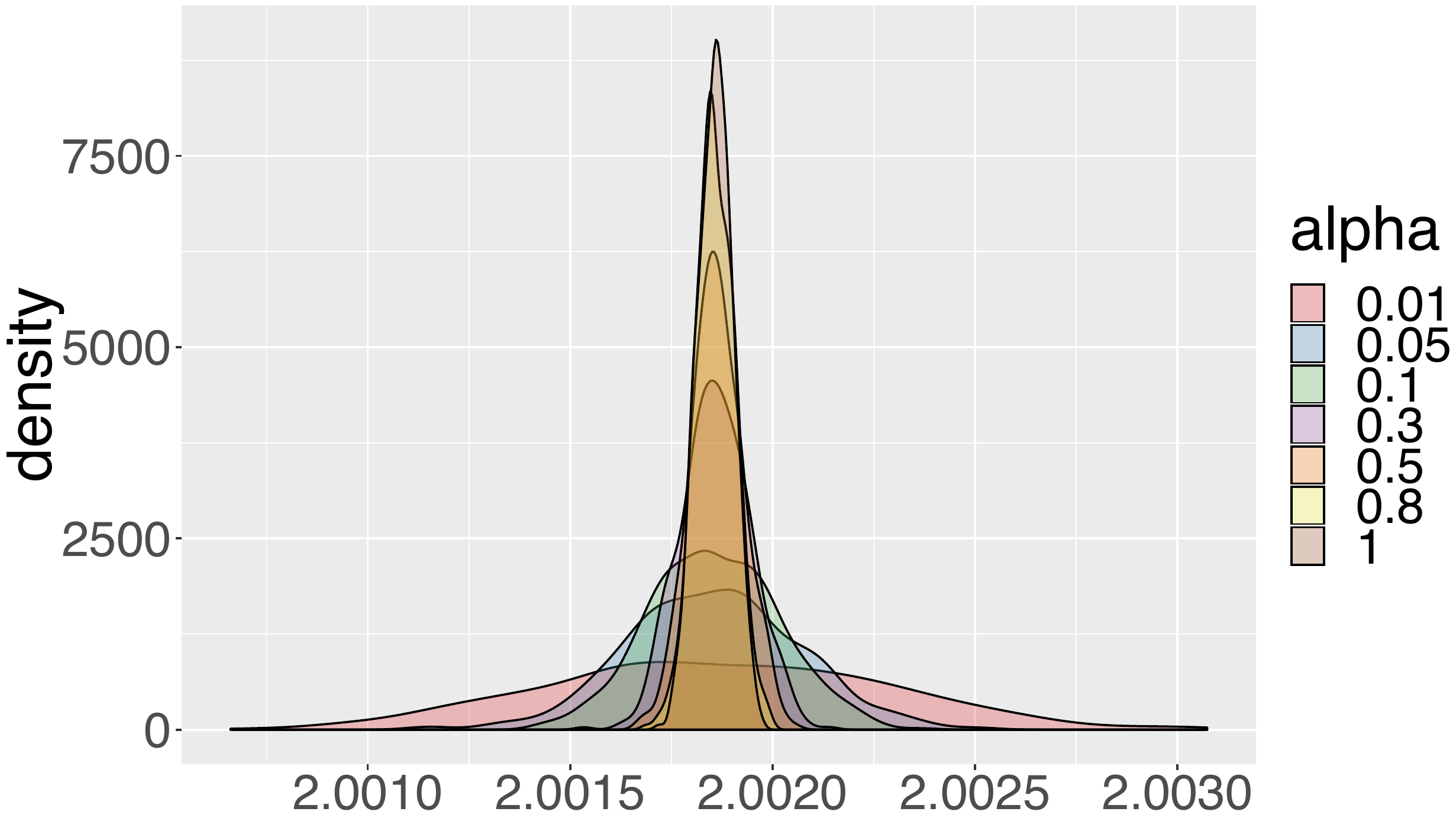} 
\caption{Estimation of the limit distribution of the estimate of $\beta_1$ after $8000$   iterations for the mini-batch algorithm with   proportions $\alpha\in\{0.01, 0.05, 0.1, 0.3, 0.5, 0.8, 1\}$.}\label{fig_frailty_limitvariance}
\end{center}
\end{figure}
 
\begin{figure}[t]
\begin{center}
\includegraphics[width=0.5\textwidth]{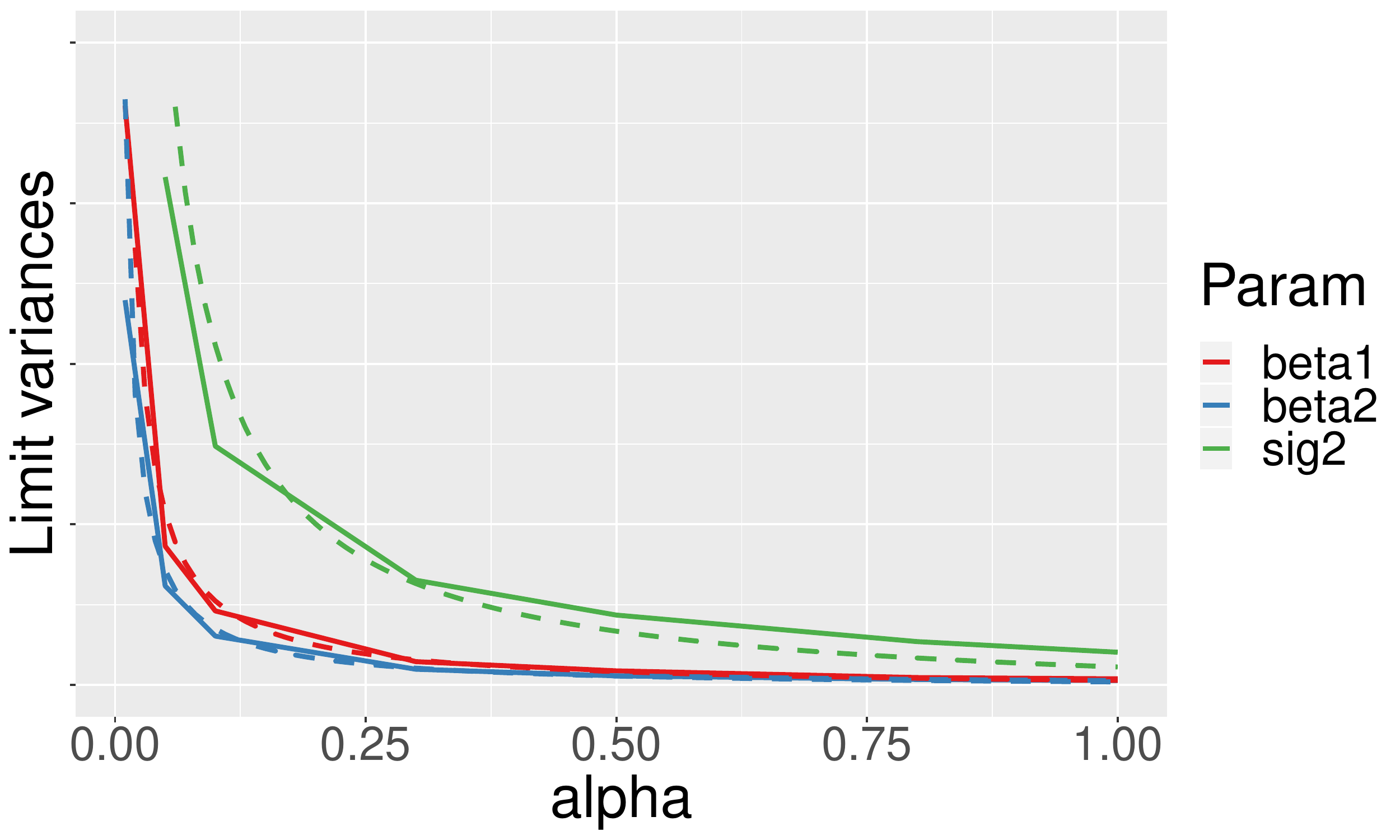} 
\caption{Sample variances of the parameter estimates  after $8000$   iterations as a function of the   mini-batch    proportion $\alpha$ (solid lines) and adjusted theoretical limit variances (dashed lines). }
\label{fig_frailty_limitvariance_theo}
\end{center}
\end{figure}

In a simulation study we consider  the frailty model with  parameters fixed to ${\boldsymbol\beta}=(\beta_1,\beta_2)=(2,3), \lambda_0=3$, $\sigma^2=2$ and $\rho=3.6$. We set $n=5000$ and $m=100$. The covariates ${\mathbf X_{ij}}$ are drawn  independently from the uniform distribution for every dataset. 
  
  In the first setting,  500 datasets are generated and    the mini-batch MCMC-SAEM algorithm with random initial values and   mini-batch proportions $\alpha$ between $0.01$ and $1$
  is applied. 
Figure \ref{fig_frailty_confint} and \ref{fig_frailty_loglik} shows the evolution of the precision of the mean of the estimates $\bar\beta_{1,k}=\sum_{l=1}^k\beta_{1,l}/k$ of parameter component  $\beta_1$ as a function
of the number of epochs.
 Figure \ref{fig_frailty_confint}   shows  80\%-confidence bands for  $\beta_1$, and Figure  \ref{fig_frailty_loglik}  gives the corresponding evolution of the logarithm of the empirical mean squared error. Again, it is clear from both graphs that the rate of convergence depends  on the mini-batch proportion. At (almost) any number of epochs, the mean squared error and the width of the confidence intervals  is increasing in  $\alpha$. From  Figure \ref{fig_frailty_confint} and \ref{fig_frailty_loglik}  we see, for instance, that
 the mini-batch version with $\alpha=0.01$ attains
  convergence   after only three epochs, while almost 30 epochs are required when $\alpha=0.1$. 
The   choice of $\alpha$ that achieves the fastest convergence is the smallest mini-batch proportion, here $0.01$. 
%

  In the second setting, we  study the asymptotic behavior of the estimates, when the algorithms are supposed to  have converged, that is after 8000 iterations. 
To evaluate the variance of the estimates that is only due to the stochasticity of the algorithm, we fix a dataset and run the different algorithms 500 times using different initial values.  
 Figure \ref{fig_frailty_limitvariance} illustrates the limit distributions of the estimates of $\beta_1$, which seem to be Gaussian, centered at the same value, but with varying variances. 
 Figure \ref{fig_frailty_limitvariance_theo} gives the corresponding values of the limit variances for the different parameter estimates.
 Again, we see that the limit variance increases when $\alpha$ decreases.  This is expected as less data are visited per
 iteration for smaller $\alpha$. Furthermore,  we conjecture in Section  \ref{subsec:limitvar} that
   the theoretical limit variance of the sequence generated by the algorithm with minibatch size $\alpha$  is equal to $
   V_1 ({2-\alpha})/\alpha $ where $ V_1$ stands for the limit variance of the batch algorithm. Figure \ref{fig_frailty_limitvariance_theo} shows a very good fit of
 the sample variances of the different parameters to 
  the  function $\alpha\mapsto v_1 ({2-\alpha})/\alpha $.

\section{Conclusion}
 
In this paper we have proposed to perform  mini-batch sampling in the MCMC-SAEM algorithm. We have shown that the mini-batch algorithm belongs to the family of classical  MCMC-SAEM algorithms with  specific  transition kernels.
It is also shown that the mini-batch algorithm converges almost surely, whenever the original algorithm does. 
Concerning the limit distribution, according to heuristic arguments and simulation results in different models, estimators are asymptotically normal and we have quantified  the impact of the mini-batch proportion onto the limit variance. However, the formal proof of this result is left for future work.

The numerical experiments carried out in various latent variable models illustrate that  mini-batch sampling  leads to an important speed-up of convergence at the beginning of the algorithm  compared to the original algorithm. 
In most papers on mini-batch sampling as well as in this paper, the illustration of this speed up  is based on a
comparison of estimates at the same number of epochs. However,  the computing time of an epoch depends on the mini-batch
proportion $\alpha$, on the amount of data used in the stochastic approximation step and on the computing time of the
M-step, which is performed much more often when $\alpha$ is small. Indeed, the smaller the value of $\alpha$, the larger the number of   M-steps within an epoch.  In the frailty model, for instance, the maximisation in the M-step is not explicit and thus more time-consuming than the other steps of the algorithm, which makes mini-batch sampling less attractive for practical use in  this model. This also raises the question of the interest of an analysis of the 
performance of the estimators with respect to the number of epochs. From a practical point of view, in future studies, we advocate that it would be much more appealing to compare algorithms relying on the same computing time rather than on the same number of epochs.

The study of the stochastic block model has shown that the common presentation of mini-batch sampling as a method where a subset of the data is selected to perform an iteration is misleading. Indeed, in the algorithm,  first the latent components that are to be simulated are chosen, and only then, the data that are associated with these updated latent components are determined to perform the update of the sufficient statistic. In models where every observation only depends  on a single latent component, the proportion of   data   used for the update equals the proportion of simulated latent components. However, in models with a more involved dependency structure as SBM this does no longer hold. 
As a consequence, in the SBM the computing time of one iteration of the SAE-step in the mini-batch
algorithm is not $\alpha$ times  the   corresponding SAE-step in the batch version.

These issues on the computing time lead us to the important question of how to make good use of mini-batch sampling in practice, where we are often confronted to constraints on the computing time. 
It seems to us that the relevant problem is to find the algorithm that  provides the best  results within some allotted computing time.
As we have seen in Section \ref{subsec:imageanalysis}, 
combining mini-batch sampling with an increase of the number of observations compared to the batch version is a means to achieve more accurate estimates under a given constraint on the computing time.
Indeed, increasing the sample size for mini-batch algorithms  compensates for the loss in accuracy of the final estimates, while the acceleration of the convergence at the beginning of the algorithm ensures the convergence of the MCMC-SAEM algorithm within the considered computing time.

To find the optimal mini-batch proportion $\alpha$ and the associated optimal sample size, an  analysis of the computing time per iteration is required, instead of  per number of epochs. These optimal values are  model dependent.
Furthermore, as any
MCMC-SAEM algorithm must always be run until convergence, it is necessary to understand the impact of the mini-batch
proportion $\alpha$ and the sample size on the convergence of the algorithm, that is, on the number of iterations required to achieve convergence.

\medskip
All programs are available on request from the authors.

\begin{acknowledgements}
Work partly supported by the grant ANR-18-CE02-0010 of the French National Research Agency ANR (project EcoNet).
\end{acknowledgements}

%

\end{document}